\documentclass[12pt,oneside,reqno,twoside]{article}

\usepackage{geometry,verbatim}
\usepackage{anyfontsize}

\usepackage{enumerate}
\usepackage[shortlabels]{enumitem}

\usepackage{
	mathtools,amsmath,amssymb,enumitem,amsthm,nccmath
}
\usepackage{graphicx,url}

\usepackage{xfrac} 
\usepackage{nicefrac}

\usepackage{color,cite,verbatim}

\usepackage[svgnames]{xcolor}

\usepackage[colorlinks=true, linkcolor=Maroon, urlcolor=Maroon]{hyperref}

\hypersetup{citecolor=blue}

\mathtoolsset{showonlyrefs} 

\usepackage{dsfont}
\usepackage{bm}
\usepackage{mathrsfs} 
\usepackage{suetterl} 
\usepackage{caption}
\usepackage{subcaption}		
\usepackage{braket} 
\usepackage{stmaryrd} 

\usepackage[T1]{fontenc}
\usepackage{currvita}

\usepackage{csquotes} 

\usepackage{empheq}
\usepackage[most]{tcolorbox}

\usepackage{float}
\usepackage[all]{xy}
\usepackage{tikz}
\usepackage{pgfplots}
\usetikzlibrary{arrows,arrows.meta,backgrounds,positioning,fit}
\usepgfplotslibrary{fillbetween}
\usetikzlibrary{patterns}
\pgfplotsset{width=7cm,compat=newest}

\tikzset{%
	dots/.style args={#1per #2}{%
		line cap=round,
		dash pattern=on 0 off #2/#1
	}
}

\allowdisplaybreaks[1]		

\def\N{{\mathbb N}}	
\def\R{{\mathbb R}}				 
\def\Z{{\mathbb Z}}				
\def\T{{\mathbb T}}
\renewcommand{\S}{\mathbb S}

\def\GLn{\mathrm{GL}}

\def\ba{{\bf a}}

\def\bh{{\bf h}}

\def\bu{{\bf u}}\def\bU{{\bf U}}
\def\bv{{\bf v}}
\def\bw{{\bf w}}
\def\bx{{\bf x}}

\def\bR{\mathbf{R}}

\def\bU{\mathbf{U}}
\def\bV{\mathbf{V}}

\def\balph{{\bm\alpha}}

\def\bgam{{\bm\gamma}}
\def\bomeg{{\bm\omega}}

\def\vep{{\varepsilon}}
\def\vphi{{\varphi}}						

\def\cA{{\mathcal A}}

\def\cD{{\mathcal D}}
\def\cE{{\mathcal E}}

\def\cJ{{\mathcal J}}					 
\def\cK{{\mathcal K}}					
\def\cL{{\mathcal L}}
\def\cM{{\mathcal M}}

\def\cR{{\mathcal R}}
\def\cS{{\mathcal S}}

\def\frakD{\mathfrak{D}}

\def\dx#1{\mathrm{d}#1\ }
\def\dxx#1{\mathrm{d}#1}

\def\1{{\bf 1}}

\newcommand{\scp}[2]{\langle#1,#2\rangle}
 
\newcommand{\floor}[1]{\left\lfloor #1 \right\rfloor}

\newcommand{\supp}{\operatorname{supp}}

\def\dist{\operatorname{dist}}

\def\diag{\operatorname{diag}}

\makeatletter
\newcommand\avsuminner[2]{%
	{\sbox0{$\m@th#1\sum$}%
		\vphantom{\usebox0}%
		\ooalign{%
			\hidewidth
			\smash{\vrule height\dimexpr\ht0+1pt\relax depth\dimexpr\dp0+1pt\relax}%
			\hidewidth\cr
			$\m@th#1\sum$\cr
		}%
	}%
}
\makeatother

\usepackage{scalerel}[2014/03/10]
\usepackage[usestackEOL]{stackengine}
\def\avint{\,\ThisStyle{\ensurestackMath{%
			\stackinset{c}{.2\LMpt}{c}{.5\LMpt}{\SavedStyle-}{\SavedStyle\phantom{\int}}}%
		\setbox0=\hbox{$\SavedStyle\int\,$}\kern-\wd0}\int}

\makeatletter
\def\namedlabel#1#2{\begingroup					
	#2%
	\def\@currentlabel{#2}%
	\phantomsection\label{#1}\endgroup
}
\makeatother

\theoremstyle{plain}
\newtheorem{theorem}{Theorem}[section]		
\newtheorem{definition}[theorem]{Definition}	
\newtheorem{assumption}{Assumption}
\newtheorem{proposition}[theorem]{Proposition}		   
\newtheorem{example}[theorem]{Example}

\newtheorem{lemma}[theorem]{Lemma}

\newtheorem{remark}[theorem]{Remark}
\newtheorem*{remark*}{Remark}

\numberwithin{equation}{section}

\def\Cper{C_{\mathrm{per}}}

\def\inter{e^{\mathrm{(inter)}}}

\def\mono{e^{\mathrm{(mono)}}}
\def\smono{\cS^{\mathrm{(mono)}}} 
\def\CB{e^{\mathrm{(CB)}}}
\def\sPK{\cS^{\mathrm{(PK)}}} 
\def\Vmono{V^{\mathrm{(mono)}}}

\def\sitepotmono{\phi^{\mathrm{(mono)}}}
\def\dsitepotmono{\Delta\Phi^{\mathrm{(mono)}}}

\def\Vinter{V^{\mathrm{(inter)}}}

\def\etot{e^{\mathrm{(tot)}}}
\def\etotCB{e^{\mathrm{(tot)}}_{\mathrm{CB}}}

\def\GSFE{\Phi^{\mathrm{(GSFE)}}}

\def\sitepotinter{\phi^{\mathrm{(inter)}}}
\def\dsitepotinter{\Delta\Phi^{\mathrm{(inter)}}}
\def\sitepotdouble{\Psi^{\mathrm{(inter)}}}

\def\demono{\Delta E^{\mathrm{(mono)}}}
\def\deinter{\Delta E^{\mathrm{(inter)}}}
\def\detot{\Delta E^{\mathrm{(tot)}}}

\def\Eapprox{E^{\mathrm{(approx)}}}

\def\mmono{m^{\mathrm{(mono)}}}
\def\Mmono{M^{\mathrm{(mono)}}}
\def\mmonoat{m^{\mathrm{(mono,at)}}}
\def\Mmonoat{M^{\mathrm{(mono,at)}}}

\def\minter{m^{\mathrm{(inter)}}}
\def\Minter{M^{\mathrm{(inter)}}}

\def\atd{a_0}        

\def\MoireCell{\Gamma_\cM}
\def\MoireCellres{\Gamma_\cM^{(0)}}

\def\MoireBZ{\Gamma_\cM^*}

\def\mavint{\avint_{\MoireCell}\hspace{-1ex}}
\def\mavresint{\avint_{\MoireCellres}\hspace{-1ex}}

\def\MoireML{\cR^*_\cM}
\def\MoireRL{\cR_\cM}
\def\MoireRLV{A_\cM}
\def\MoireMLV{B_\cM}

\def\MoireR{R_\cM}

\def\Lloc{L_{\mathrm{loc}}}
\def\Lper{L_{\mathrm{per}}}
\def\Wper{W_{\mathrm{per}}}
\def\Wbper{\bar{W}_{\mathrm{per}}}

\def\bueq{\bu^{\mathrm{(eq)}}}		

\def\ueq{u^{\mathrm{(eq)}}}

\def\buCB{\bu^{\mathrm{(CB)}}}

\def\uCB{u^{\mathrm{(CB)}}}

\def\bUCB{\bU^{\mathrm{(CB)}}}
\def\bdu{{\bm \delta \bu}}
\def\du{\delta u}
\def\bdv{{\bm \delta \bv}}
\def\dv{\delta v}
\def\dw{\delta w}

\newcommand{\delj}{\Delta^{(j)}}
\newcommand{\delji}{\Delta^{(j,3-j)}}

\newcommand{\delCB}{\delta^{\mathrm{(CB)}}}
\newcommand{\delmono}{\delta^{\mathrm{(mono)}}}

\def\bot{D_{1\to 2}}

\def\disregj{D_{j\to3-j}}

\def\mfrac#1{\left\{#1\right\}_\cM}
\def\Lfrac#1#2{\left\{#1\right\}_{#2}}
\def\mfloor#1{\floor{#1}_\cM}
\def\Lfloor#1#2{\floor{#1}_{#2}}

\usepackage{color, colortbl}

\definecolor{UTorange}{RGB}{180,100,0} 
\definecolor{UTgrey}{RGB}{51, 63, 72}

\definecolor{UMNmaroon}{RGB}{122,0,25} 
\definecolor{UMNgold}{RGB}{255,204,51}

\definecolor{primary}{RGB}{122,0,25}
\definecolor{secondary}{RGB}{255,204,51}

\definecolor{CBred}{RGB}{218, 85, 38}
\definecolor{CBorange}{RGB}{246, 137, 61}
\definecolor{CByellow}{RGB}{254, 188, 56}
\definecolor{CBbeige}{RGB}{216, 198, 180}
\definecolor{CBgray}{RGB}{105, 127, 144}
\definecolor{CBgreen}{HTML}{5BA300}
\definecolor{CBblue}{HTML}{0073E6}

\makeatother

\makeatletter
\def\@tocline#1#2#3#4#5#6#7{\relax
	\ifnum #1>\c@tocdepth 
	\else
	\par \addpenalty\@secpenalty\addvspace{#2}%
	\begingroup \hyphenpenalty\@M
	\@ifempty{#4}{%
		\@tempdima\csname r@tocindent\number#1\endcsname\relax
	}{%
		\@tempdima#4\relax
	}%
	\parindent\z@ \leftskip#3\relax \advance\leftskip\@tempdima\relax
	\rightskip\@pnumwidth plus4em \parfillskip-\@pnumwidth
	#5\leavevmode\hskip-\@tempdima
	\ifcase #1
	\or\or \hskip 2em \or \hskip 3em \else \hskip 4em \fi%
	#6\nobreak\relax
	\hfill\hbox to\@pnumwidth{\@tocpagenum{#7}}\par
	\nobreak
	\endgroup
	\fi}
\makeatother

\usepackage{fancyhdr}

\fancypagestyle{headings}{%
	\cfoot{\thepage}
}
\usepackage{orcidlink}
\usepackage[affil-it]{authblk}

\begin{document} 
	
	\date{\today}
	\title
    {Mathematical foundations of phonons in incommensurate materials}
	
	\author[1]{Michael Hott\orcidlink{0000-0003-4243-6585}%
		\thanks{E-mail: \texttt{mhott@umn.edu}}}
 
   \author[1]{Alexander B. Watson\orcidlink{0000-0002-7566-4851}%
		\thanks{E-mail: \texttt{abwatson@umn.edu}}}
 
	\author[1]{Mitchell Luskin\orcidlink{0000-0003-1981-199X}%
		\thanks{E-mail: \texttt{luskin@umn.edu}}}

\affil[1]{School of Mathematics\\ University of Minneapolis, Twin Cities\\ Minneapolis, MN 55455, USA}

	
	
	\newgeometry{margin=1in}
	
	\maketitle

	\begin{abstract}
        The physical observation that it is energetically favorable for interacting particles to arrange in a periodic fashion remains a largely open problem, referred to as the \emph{crystallization conjecture}. In specific models, e.g., coming from atomic pair potentials, periodic configurations can be shown to be stable under, both, global $\ell^2$ or local perturbations. The situation changes when one studies aperiodic media. The specific class of aperiodic media we are interested in arise from taking two 2D periodic crystals and stacking them parallel at a relative twist. We study different notions of stability for such systems. The goal of our analysis is to provide phonons in such systems with meaning. In periodic media, phonons are generalized eigenvectors for a stability operator acting on $\ell^2$, coming from a mechanical energy. Using ideas rigorously established for the 1D Frenkel-Kontorova model, we assume that we can parametrize minimizing lattice deformations w.r.t. local perturbations via continuous \emph{stacking-periodic} functions, see \cite{Cazeaux-Massatt-Luskin-ARMA2020}, for which we previously derived a continuous energy density functional \cite{hott2023incommensurate}. Such (continuous) energy densities are analytically and computationally much better accessible compared to discrete energy functionals. In order to pass to an $\ell^2$-based energy functional, we also study the offset energy w.r.t. given lattice deformations, under $\ell^1$-perturbations. Our findings show that, in the case of an undeformed bilayer heterostructure, while the energy density can be shown to be stable under the assumption of stability of individual layers, the offset energy fails to be stable in the case of twisted bilayer graphene. We then establish conditions for stability and instability of the offset energy w.r.t. the relaxed lattice. Finally, we show that, in the case of incommensurate bilayer homostructures, i.e., two equal layers, if we choose minimizing deformations according to the global energy density above, the offset energy is stable in the limit of zero twist angle. Consequently, in this case, one can then define phonons as generalized eigenvectors w.r.t. the stability operator associated with the offset energy.

	\end{abstract}

	\paragraph{Statements and Declarations} The authors do not declare financial or non-financial interests that are directly or indirectly related to the work submitted for publication.
	
	\paragraph{Data availablity} The manuscript has no associated data.
	
	\paragraph{Acknowledgments} MH's and ML's research was partially supported by Simons Targeted Grant Award No. 896630. AW's, and ML's research was supported in part by NSF DMREF Award No. 1922165. ML's research was also supported in part by grant NSF PHY-2309135 to the Kavli Institute for Theoretical Physics (KITP). AW's research was also supported in part by grant NSF DMS-2406981. The authors would like to thank Ziyan (Zoe) Zhu for helpful comments and stimulating discussions.
	
	\tableofcontents
	
	\restoregeometry 

    \section{Introduction}

	Generally, phonons are understood as generalized eigenvectors of a stability operator acting on $\ell^2$, associated with a mechanical energy. More precisely, expanding a mechanical energy functional w.r.t. equilibrium positions $\{R_{\mathrm{eq}}^j\}_{j\in\cJ}$, perturbed by lattice deformations $\du:\{R_{\mathrm{eq}}^j\}_{j\in\cJ}\to \R^d$,
    \begin{align}
        E(\{R_{\mathrm{eq}}^j+\du(R_{\mathrm{eq}}^j)\}_{j\in\cJ}) \, = \, E(\{R_{\mathrm{eq}}^j\}_{j\in\cJ}) \, + \, \delta^2E(\{R_{\mathrm{eq}}^j\}_{j\in\cJ})[\du,\du] \, + \, \ldots \, ,
    \end{align}
    we define phonons as generalized eigenvectors of the stability operator $\delta^2E(\{R_{\mathrm{eq}}^j\}_{j\in\cJ})$ acting on $\ell^2$. Crucially, this requires to define an energy functional yielding a well-defined stability operator acting on $\ell^2$. Notice that $\delta^2E(\{R_{\mathrm{eq}}^j\}_{j\in\cJ})$ is necessarily positive (semi-)definite.
    
    \par In the case of periodic media described by a (multi)lattice, the unperturbed lattice positions can be shown to be stable under global $\ell^2$ or local perturbations. An energy functional, in the simplest case, is given by the sum of harmonic springs. Then the associated phonon Hamiltonian can be computed as the root of the second variation of the mechanical energy. This is due to the fact that, in this case, we model lattice vibrations via harmonic oscillators at every lattice position. Due to the lattice periodicity, Bloch-Floquet theory enables us to analyze the spectral properties of the associated phonons. More precisely, in this case, we define the \emph{phonon Hamiltonian} as the (square root of) the stability operator associated with the energy functional. 
	\par In general aperiodic media, however, the story is not as straight-forward. First, we need to study a mechanical energy to obtain a relaxed lattice configuration. This relaxed lattice configuration will generally not be periodic, so Bloch-Floquet theory is not applicable. Moreover, while it is possible to again model the mechanical energy via harmonic springs, analytically or computationally determining the associated energy minimum is generally challenging. Our goal is thus twofold: 
	\begin{enumerate}
		\item Determine an analytically/computationally accessible mechanical energy to compute an energy minimizer.
		\item Find an appropriate phonon model  that allows us to treat phonons as $\ell^2$-waves.
	\end{enumerate}

    \par In order to analyze this question, we shall turn to incommensurate heterostructures. These arise when two or more two-dimensional crystal layers are stacked on top of each other at a relative twist, forming a quasiperiodic medium. Due to the periodicity of the individual layers, quasiperiodic macroscopic patterns, referred to as \emph{moir\'e pattern}, emerge. Such systems have gained a lot of attention, when it turned out that, at certain small twist angles, known as \emph{magic angles}, special electronic properties are unlocked. More precisely, Bistritzer and MacDonald \cite{bistritzer2011moire} derived a reduced model for twisted bilayer graphene that predicted the occurrence of nearly flat electronic bands at a twist angle of $\sim1.1^\circ$, which in turn yields a highly localized density of states. As a consequence of the flat dispersion, the interaction between electrons become dominant and various behavior typical for strongly correlated systems has been observed to emerge, see \cite{CaoJarilloHerero2018}. One of the mysterious phases is the superconducting phase, for it occurs at a very low electron density compared to the transition temperature, which is rather typical for unconventional superconductors. However, the mechanism underlying the superconducting phase remains a puzzle. 
    \par Spectroscopic measurements \cite{oh2021unconventionalSC} seem to hint at unconventional superconductivity, while electronic transport \cite{liu2021coulombscreening} shows independence of the critical temperature w.r.t. Coulomb screening. This, in turn, suggests that phonons may play a critical role in understanding superconductivity.
    
    \par One of the main obstacles in understanding the role of phonons, is the large number of -- even intersecting -- phonon bands. Physicists have thus come up with certain selection rules to pick only a few bands and projected them onto a low-energy regime, similar as previously done for electrons by Bistritzer and MacDonald, see \cite{koshinonam2019phonons,zoe-moirephonon2022low}. In addition, these computations cutoff the phonon energy at zero and ignore possible phonon bands with negative dispersion, see, e.g., \cite{zhu2022lowenergyphonon,zoe-2024microscopic}. Our goal is to provide a mathematical foundation for phonons in incommensurate multilayer structures, and to study whether or not this negative dispersion is a model relic or if it could be avoided by adjusting the model.

    \par In order to find an energy functional, we can generally study (at least) two models
    \begin{enumerate}
        \item Minimize a \emph{global} energy w.r.t. $\ell^2$-perturbations.
        \item Minimize a \emph{local} energy in any finite subset of the medium.
    \end{enumerate}
    The first approach has the advantage that we would like to model phonons as $\ell^2$-normalized waves propagating through an aperiodic medium. However, the main disadvantage here is that, in order to construct a relaxed configuration, one typically constrains the problem to a supercell, corresponding, e.g., to a moir\'e cell, and chooses, e.g., periodic boundary conditions, see \cite{malena2018,malena2023,Srolovitz2015}. This is due to the fact that displacement functions are parametrized by their discrete lattice positions and continuum approximations that are commonly used for periodic systems, such as the Cauchy-Born approximation, are not accessible to the generic aperiodicity. Computations over large supercells are challenging, as they scale with $O_{\theta\to0}(\frac1{\theta^2})$ for twist angles $\theta$. 
    \par In order to minimize local energies in a quasiperiodic medium, Aubry and Le Daeron \cite{aubry1983frenkel-kontorova} studied the Frenkel-Kontorova model. In this model, lattice displacements on $\Z$ are modeled via harmonic springs, and an external incommensurate potential is added. Then, under the assumption of an additonal recurrence condition, it was shown that a minimizer can be parametrized by a periodic function with the same periodicity as the external potential. This  \emph{hull} function, see also \cite{bellissard-quasi-per1982,bellissard1994noncommutative,bellissard2002coherent}, is either continuous, when the external potential is very weak, or discontinuous w.r.t. a Cantorus when the external potential is very strong. Analogous results were obtained for 1D rippling chains \cite{cazeauxrippling2017} by applying Aubry-Mather theory.
    \par Motivated by these 1D results, Cazeaux, Massatt and one of the authors \cite{Cazeaux-Massatt-Luskin-ARMA2020}, see also \cite{Carr2018}, developed an energy density for continuous functions that are periodic w.r.t. to the local \emph{stacking configuration}. We shall note that establishing a two- and higher-dimensional Aubry-Mather theory remains an open problem. In particular, this approach introduces a \emph{third} model for an energy functional:
    \begin{enumerate}[3.]
    	\item Minimize an energy density w.r.t. \emph{stacking-periodic} displacements.
    \end{enumerate}
    From the previous discussion, we have that such a model replaces the second model above in the case of the 1D Frenkel-Kontorova model, while, in 2D, it is a distinct model.
    \par The main advantage of this approach is that, instead of sampling over supercells, it suffices to sample over stacking configurations, making the minimization problem analytically and computationally more accessible; the predicted results are in agreement with the experimental observations \cite{cazeauxclark23domainwall}. However, since this energy functional is derived for energy minimizers only, and since these are assumed to be periodic, and thus non-decaying in $\ell^2$, we cannot treat phonons in the same framework.

    \par In order to pass from a model based on stacking-periodic perturbations to $\ell^2$-perturbations, we employ an idea pursued in \cite{hainzl2007QEDmean-no-photon}. There, a theory is developed to study energy minimizers for interacting Dirac particles. They start by constructing a \emph{translation-invariant} minimizer in the absence of an electric background field for which the energy \emph{density} is stable w.r.t. \emph{translation-invariant perturbations}. Using this minimizer as a reference state, they then consider the \emph{offset energy} and minimize it w.r.t. Hilbert-Schmidt perturbations. These two minimizers are shown to coincide.

    \par In this work, we adopt these ideas in the following way: We start by recalling results established \cite{hott2023incommensurate}, where we generalized the ideas in \cite{Cazeaux-Massatt-Luskin-ARMA2020} to obtain an atomistic energy based on displacements that are \emph{periodic} w.r.t. the local stacking configuration. Crucially, this energy functional is obtained via a thermodynamic limit, analogously to \cite{hainzl2007QEDmean-no-photon} 
    \begin{equation}
        \etot(\bu) \, := \, \lim_{r\to\infty}\frac{E\big|_{B_r}(\bu)}{|B_r|} \, .
    \end{equation}
    We also establish conditions on the many-body potentials that allow us to employ the results in \cite{hott2023incommensurate}. For more details, we refer to Section \ref{sec-assumptions} for the assumptions on the potentials, and to Section \ref{sec-energy-well-def} to show well-definedness of the energy functional. 
    \par In Section \ref{sec-relax}, we continue to prove the existence a locally unique minimizer $\bueq$ of $\etot$ which is stable w.r.t. perturbations that are periodic w.r.t. the local environment. This minimizer lets us compute the \emph{offset energy} w.r.t. the total energy of $\bueq$
    \begin{equation}
        \detot_{\bueq}(\bdu) \, := \, \lim_{r\to\infty}\Big(E\big|_{B_r}(\bueq+\bdu)-E\big|_{B_r}(\bueq)\Big) \, .
    \end{equation}

    \paragraph{Definition of phonon Hamiltonian} A natural question to ask is whether $\bueq$ is still a minimizer for the offset energy, i.e., whether $0$ is a minimizer of $\detot_{\bueq}$. Assuming sufficient regularity of the involved potentials, expanding in orders of $\bdu$ yields that an appropriate space to study $\detot$ is given by $\ell^1$, see \ref{sec-first-var} below. 
    
    \par The goal of this work is thus to study when the minimization problem associated with $\etot$ is consistent with that associated with $\detot_{\bueq}$.  When it is, we can rigorously pass from the configuration space framework to the $\ell^2$-phonon framework, and the phonon Hamiltonian is defined as 
    \begin{equation}\label{def-phonon-Hamiltonian}
        H_{\mathrm{ph}}:=\sqrt{\delta^2\detot_{\bueq}({\bf 0})}:\ell^2\to\ell^2 \, .    
    \end{equation}
    As we will see, we cannot employ standard perturbative methods and, instead, will provide sufficient criteria to ensure or exclude phonon stability. These criteria need to be numerically verified. 
    
    \par We prove in Section \ref{sec-first-var} that
    \begin{equation}
        \delta^{\Wbper^{1,2}} \etot (\bueq) \, = \, 0 \quad \Rightarrow \quad \delta^{\ell^1}\detot_{\bueq}({\bf 0}) \, = \, 0 \, ,
    \end{equation}
    i.e., that a critical point of $\etot_{\bueq}$ under moir\'e-periodic perturbations remains a critical point of $\detot$ under $\ell^1$-perturbations. In particular, this means that $\delta^2\detot_{\bueq}({\bf 0}):(\ell^2)^2\to\R$ is well-defined. Equivalently, by duality, $\delta^2\detot_{\bueq}({\bf 0})\ell^2\to\ell^2$ is well-defined. We refer to the Section \ref{sec-notation} for the chosen notation, and to the respective sections for the precise statements. 
    \par The next question thus is whether stability of $\etot$ is inherited by $\detot_{\bueq}$. To answer this question, we start in Section \ref{sec-phonon-stab} by proving that, for certain choices of interlayer pair potentials
    \begin{equation}
        \delta^2\etot({\bf 0}) \, \geq \, \kappa \, > \,  0 \quad \wedge \quad \delta^2\big|_{\bdu={\bf 0}}\detot_{\bf 0}(\bdu) \, < \, 0 \, .
    \end{equation}
    This means that the undeformed union of lattices is stable under periodic perturbations, while it is not under $\ell^2$-perturbations. It matches the intuition that we expect the undeformed lattice not to be stable. The condition $\delta^2\etot({\bf 0}) \geq \kappa$ is typically referred to as \emph{stability of the lattice}. In the case of a single layer, this has been shown to be a necessary condition for stability of minimizers, see \cite[Appendix C]{ortnertheil2013}. Including the relaxed lattice positions, we prove that, in the case of pair interlayer potentials $v$,
    \begin{align}
        \MoveEqLeft\inf_{\|\bdv\|_{\ell^2}=1}\delta^2\big|_{\bdu={\bf 0}}\detot_{\bueq}(\bdu)[\bdv,\bdv] \,  \\ 
        \leq \, &\inf_{a\in\S^2}\int_{\R^2} \dx{x} \partial_a^2 v\big(x+\ueq_2(\MoireRLV A_1^{-1}x)-\ueq_1(\MoireRLV A_2^{-1}x)\big) \, .
    \end{align}
    Thus, a sufficient condition for instability is a negative r.h.s. This instability criterion stems from constructing an approximate minimizer for which the monolayer phonon energy vanishes while the interlayer phonon energy is minimized. 
    \par Conversely, a sufficient condition for phonon stability is given by
    \begin{align}\label{eq-stab-crit-intro}
			\sum_{R_1\in\cR_1} D^2v(x-R_1+\ueq_2(\MoireRLV A_1^{-1}x)-\ueq_1(R_1+\MoireRLV A_2^{-1}x) ) \, \geq \, 0 
    \end{align}
    for all $x\in\Gamma_1$, in the sense of quadratic forms. Here, $\cR_1$ refers to the lattice of layer 1, and $\Gamma_1$ is the corresponding unit cell. More precisely, we show in Section \ref{sec-phonon-stab} that this condition implies positive (semi-)definiteness of  $\delta^2\detot_{\bueq}$. While we explain the ideas only for interlayer pair potentials, we believe that one can generalize them to, more general, empirical many-body potentials which have become computationally more accessible and which promise to yield yield highly accurate calculations, thanks to the advances made in machine-learned empirical many-body, see, e.g., \cite{chen2022qm,ortner21,shapeev16}.
    \par Relating the present energy functional to an appropriate Allen-Cahn functional, we argue in Section \ref{sec-asymp-analysis} that, for a purely twisted bilayer system for which the layer coupling is mediated through a pair potential $v$, at (twist angles)$\to0$, the Allen-Cahn minimizers converges to the potential wells of the misfit energy $\Phi_0(x):=\sum_{n\in\Z^2}\frac{v(A(x-n))}{\det A}$. For the Allen-Cahn theory to be applicable, we need to require that  at the potential wells $x_{\mathrm{AB}}$, $x_{\mathrm{BA}}$, we have that
    \begin{align}\label{eq-AC-stab-intro}
        D^2\Phi_0(x_{\mathrm{BA}/\mathrm{AB}}) \, > \, 0 \, .
    \end{align}
    This, in turn, means that we need to establish strict inequality in \eqref{eq-stab-crit-intro}. We establish \eqref{eq-AC-stab-intro} for 2D pair potentials for graphene, at the equilibrium distance $d_{\mathrm{BG}}=3.35\mbox{\ \AA}$ between the layers. Our choices of pair potentials are Lennard-Jones, Morse and Kolmogorov-Crespi potentials \cite{KolmogorovCrespi}, see Section \ref{sec-model-eval}. Consequently, in the limit of zero twist angle and for sufficiently small lattice constants, in the sense of a continuum limit, the phonon Hamiltonian $H_{\mathrm{ph}}\geq0$, see \eqref{def-phonon-Hamiltonian}, is well-defined.

	\section{Preliminaries and notation\label{sec-notation}}
	
	\subsection{Lattices}
	
	For $j=1,2$, let $A_j\in\GLn_2(\R)$ be the matrix with columns consisting of the primitive translation vectors associated with layer $j$,
 
	with layer unit cells $\Gamma_j:=A_j[0,1)^2$. We define the direct lattices and their truncations
	\begin{align}\label{eq-primitive-lattice}
		\cR_j \, := \, A_j\Z^2 \, , \quad  \cR_j(N):=A_j(\Z^2\cap[-N,N]^2).
	\end{align}
	Moreover, we define the reciprocal primitive translation vectors and the associated reciprocal lattice
	\begin{equation} 
		B_j \, := \, 2\pi A_j^{-T} \, , \quad \cR_j^* \, := \, B_j\Z^2 \, .
	\end{equation}
	Next, we define the \emph{primitive moir\'e translation vectors} and the corresponding reciprocal vectors 
	\begin{equation}
		\MoireRLV \, := \, (A_1^{-1}-A_2^{-1})^{-1} \, , \quad \MoireMLV \, := \, 2\pi \MoireRLV^{-T} \, = \, B_1 \, - \, B_2 \, .
	\end{equation}
	Then the corresponding direct and reciprocal lattices are given by
	\begin{align} \label{def-moire-lattices}
		\MoireRL \, := \, \MoireRLV \Z^2 , \quad \MoireML \, := \, \MoireMLV \Z^2 \, ,
	\end{align}
	and they have the unit cells
	\begin{align} \label{def-moirecell}
		\MoireCell \, := \, \MoireRLV[0,1)^2 \, , \quad \MoireBZ \, := \, \MoireMLV\left[-1/2,1/2\right)^2 \, .
	\end{align}
	
	\begin{assumption}\label{ass-incomm}
		$(\cR_1,\cR_2)$ is \emph{incommensurate}, i.e., 
		\begin{align} \label{eq-incomm}
			G_1 \, + \, G_2 \, = \, 0 \mbox{\ \emph{for}\ } G_j\in \cR_j^* \quad \mbox{\emph{iff}} \quad G_1=G_2=0 \, .
		\end{align}
	\end{assumption}
	
	\paragraph{Local disregistry} An important concept that allows us to analyze the quasiperiodicity of twisted bi- and multilayer structures is the following notion of a local environment, which we will refer to as \emph{(local) disregistry}. For $x\in\R^2$, define 
	\begin{equation} \label{def-mfrac}
		\mfrac{x} \, := \, \sum_{R_\cM\in\MoireRL} (x-R_\cM)\mathds{1}_{\MoireCell}(x-R_\cM)
	\end{equation}
	and
	\begin{equation}\label{def-mfloor}
		\mfloor{x} \, := \, x \, - \, \mfrac{x} \, .
	\end{equation}
	Observe that
	\begin{equation}
		\mathds{1}_{\MoireCell}(x-R_\cM)=0 \quad \Leftrightarrow  \quad x\in \MoireCell+R_\cM \, .
	\end{equation}
	Analogously, we define the decomposition w.r.t. the lattices $\cR_j$
	\begin{equation}\label{def-Lfrac-Lfloor}
		x \, = \, \Lfloor{x}{j} \, + \, \Lfrac{x}{j}
	\end{equation}
	with $\Lfloor{x}{j}\in\cR_j$, $\Lfrac{x}{j}\in \Gamma_j$. We refer to $\Lfrac{x}{k}$, $k\in\{1,2,\cM\}$ as the (local disregistry) w.r.t. layer $k$, if $k=1,2$, and w.r.t. the moir\'e lattice if $k=\cM$. 
	
	\par In order to connect these different notions of disregistry, we also introduce the \emph{disregistry matrices}
	\begin{align}\label{def-disregj}
		\disregj \, := \, I - A_{3-j}A_j^{-1}  \, .
	\end{align}

	\par We adopt \cite[Lemma 3.1]{hott2023incommensurate}. 
	
	\begin{lemma}\label{lem-disregj-calc} The disregistry matrices have the following properties.
		\begin{enumerate}[label=\textnormal{(\arabic*)}]
			\item $\disregj \, = \, (-1)^jA_{3-j}\MoireRLV^{-1}$, \label{itm-disregj}
            \item $\disregj: \MoireCell\to(-1)^j\Gamma_{3-j}$ is a bijection,
			\item $\disregj: \MoireRL\to\cR_{3-j}$ is an isomorphism,
			\item $\disregj^T:\cR_{3-j}^*\to\MoireML$ is an isomporhism,
			\item\label{itm-disregj-transf} $\Lfrac{R_j}{3-j}=\disregj\mfrac{R_j}+A_2\begin{pmatrix}
				1\\1
			\end{pmatrix}\delta_{j,1}$ for all $R_j\in\cR_j$.
		\end{enumerate}
	\end{lemma}
	Notice that \ref{itm-disregj-transf} provides a linear transformation between the different notions of disregistry. However, it cannot be extended to general $x\in\R^2$, as we have previously discussed in \cite{hott2023incommensurate}. Instead, the proper generalization of \ref{itm-disregj-transf} is given by an appropriate ergodic theorem which we will state below.
	
	\paragraph{Sublattice degrees of freedom}
	
	In addition to the primitive lattice vectors introduced above, we also now introduce \emph{basis vectors}. These allows us to include various degrees of freedom in order to describe general multilattices, e.g., the honeycomb structure of graphene, or to include higher orbitals.
	\par In particular, for $j=1,2$, let $\cA_j$ be finite sets, and to each $\alpha_j\in\cA_j$ assign a shift vector $\tau_j^{(\alpha_j)}\in\R^2\times\{0\}$. In addition, let $\gamma_j\in\Gamma_j$ be the position of the origin of layer $j$. We identify $\gamma_j$ with its embedding $(\gamma_j,0)$ into $\R^3$.
	\par Analogously to \eqref{eq-primitive-lattice}, we define
	\begin{align}
		\Omega_j \, := \, \cR_j\times\cA_j \, , \quad \Omega_j(N) \, := \, \cR_j(\lfloor |\Gamma_{3-j}|^{\frac12}N\rfloor)\times\cA_j \, .
	\end{align}
	In addition, denoting by $\sqcup$ the disjoint union, we also abbreviate
	\begin{equation}
		\Omega \, := \, \Omega_1\sqcup \Omega_2 \, , \quad \Omega(N) \, := \, \Omega_1(N)\sqcup \Omega_2(N) \, .
	\end{equation}
	For convenience, we also define 
	\begin{equation}
		\omega_j+x \, := \, (R_j+x,\alpha_j) \quad \forall \omega_j=(R_j,\alpha_j)\in\Omega_j, \, \forall x\in\R^2 \, .
	\end{equation}
	To define the \emph{relaxed} lattice positions, let $\Cper(\MoireCell;\R^3)$ denote the space of continuous periodic $\R^2$-valued functions defined on $\MoireCell$, and let $u_j\in \Cper(\MoireCell;\R^3)$. Then we assume that the total lattice positions (in Lagrangian coordinates) are given by
	\begin{equation}
		Y_j(\omega_j+\gamma_j) \, = \, \gamma_j \, + \, R_j \, + \, \tau_j^{(\alpha_j)} \, + \, u_j(\omega_j+\gamma_j) \, .
	\end{equation}
	We extend these to all $x\in\R^2$ by
	\begin{equation}
		Y_j(x,\alpha_j) \, = \, x \, + \, \tau_j^{(\alpha_j)} \, + \, u_j(x,\alpha_j) \, .
	\end{equation}
	
	\paragraph{Multivector notation} We abbreviate
	\begin{equation}
		\bomeg_j^{\ell} \, := \, (\omega_j^{(1)},\ldots,\omega_j^{(\ell)}), \quad \bR_j^{\ell} \, := \, (R_j^{(1)},\ldots,R_j^{(\ell)})    
	\end{equation}
	where $\omega_j^{(n)}=(R_j^{(n)},\alpha_j^{(n)})\in\Omega_j$. Throughout this work, we also abbreviate
	\begin{align}
		\bu \, = \, (u_j)_{j=1}^2 , \quad \bdu \, = \, (\du_j)_{j=1}^2 , \quad \bgam=(\gamma_j)_{j=1}^2, \quad \balph=(\alpha_j)_{j=1}^2 \, .
	\end{align}

	\subsection{Mechanical energy}

    We now introduce the formal expressions for the energy density, as well as for the offset energy. In the next section, we will state sufficient conditions to ensure the existence of the respective functionals. We will also prove the convergence to the respective limit functionals.
 
	\par In the following and throughout this work, let $j\in\{1,2\}$ if not specified otherwise.
    
    \subsubsection{Monolayer potential}
	
	Let
	\begin{equation}
		\Vmono_{j,\alpha_j} \, : \, (\R^3)^{\Omega_j} \to\R 
	\end{equation}
	with $\alpha_j\in\cA_j$ be given. In addition, we abbreviate the \emph{(relative) finite difference stencils} 
    \begin{align}
        \Delta_j[u_j](x,\alpha_j) \, := \, \big(u_j(\omega_j+x) - u_j(x,\alpha_j)\big)_{\omega_j\in\Omega_j} \, , \label{def-fin-diff-stenc-mono}
    \end{align}
    For $\gamma_j\in\Gamma_j$ and $N\in\N$, we then introduce the monolayer energy density
    \begin{align}\label{def-monolayer}
		\begin{split}
			\mono_{j,N,\gamma_j}(u_j) \, :=\, \frac{|\cA_j|}{|\Gamma_j||\Omega_j(N)|}\sum_{\omega_j\in \Omega_j(N)} \Vmono_{j,\alpha_j}\big(\Delta_j[u_j](\omega_j+\gamma_j)\big) \, ,
		\end{split}
	\end{align}
	where 
	\begin{equation}
		|\bigcup_{\omega_j\in\Omega_j(N)}(\Gamma_j+R_j+\tau_j^{(\alpha_j)})| \, = \, |\Gamma_1||\Gamma_2|(2N+1)^2+O_{N\to\infty}(N) \, = \, \frac{|\Gamma_j||\Omega_j(N)|}{|\cA_j|}+O_{N\to\infty}(N)
	\end{equation}
	is the area of the corresponding truncated region covered by the nuclei in layer $j$. 
    \par We formally also define  the limit energy density
    \begin{align}\label{def-monoj}
		\mono_j(u_j) \, := \, \frac1{|\Gamma_j|}\sum_{\alpha_j\in\cA_j} \mavint \dx{x}  \Vmono_{j,\alpha_j}\big(\Delta_j[u_j](x,\alpha_j)\big) \, ,
	\end{align}
 which we will prove to coincide with $\lim_{N\to\infty}\mono_{j,N,\gamma_j}(u_j)$.
    In this context, it is useful to also abbreviate the site potential density
    \begin{align}
        \sitepotmono_{j}[u_j]: \R^2\to\R, \, x\mapsto \frac1{|\Gamma_j|}\sum_{\alpha_j\in\cA_j} \Vmono_{j,\alpha_j}\big(\Delta_j[u_j](x,\alpha_j)\big) \, , \label{def-sitepotmono}
    \end{align}
    which, in turn, allows us to write
    \begin{align}
        \mono_j(u_j) \, = \, \mavint \dx{x}\sitepotmono_{j}[u_j](x) \, .
    \end{align}
    Similarly, we define the offset site potential
    \begin{align}
        \MoveEqLeft\dsitepotmono_{j}[\du_j;u_j] \, := \, |\Gamma_j|\big(\sitepotmono_{j}[u_j+\du_j]-\sitepotmono_{j}[u_j]\big)\\
        &= \, \sum_{\alpha_j\in\cA_j} \Big(\Vmono_{j,\alpha_j}\big(\Delta_j[u_j+\du_j](\cdot,\alpha_j)\big)-\Vmono_{j,\alpha_j}\big(\Delta_j[u_j](\cdot,\alpha_j)\big)\Big)\, , \label{def-sitepotmono-offset}
    \end{align}
    and the \emph{offset energy}
    \begin{align}
        \demono_{j;u_j}(\du_j) \, := \, \sum_{R_j\in\cR_j}\dsitepotmono_j[\du_j;u_j](R_j) \, .
    \end{align}
	
	\subsubsection{Interlayer coupling}
	\begin{definition}\label{def-translation-inv}
		A function 
		\begin{equation}
			V \, : \, (\R^3)^{\Omega_j} \to\R 
		\end{equation} 
		is \emph{translation-invariant} iff for any $(a_{\omega_j})_{\omega_j\in\Omega_j}\in (\R^3)^{\Omega_j}$ and any $R_j\in \cR_j$, we have that
		\begin{equation} V\big((a_{\omega_j+R_j})_{\omega_j\in\Omega_j}\big) \, = \, V\big((a_{\omega_j})_{\omega_j\in\Omega_j}\big) \, . 
		\end{equation}
		In addition, $V$ is \emph{even} iff $V(x)=V(-x)$.
	\end{definition}
	
	Let
	\begin{equation}
		\Vinter_{j,\alpha_j} \, : \, (\R^3)^{\Omega_{3-j}} \to\R \, ,
	\end{equation}
	with $\alpha_j\in\cA_j$ be even and translation-invariant. As a special instance, let $v_{\balph}:\R^3\to\R$, with $\alpha_k\in\cA_k$, $k=1,2$, $\balph:=(\alpha_1,\alpha_2)$, denote an even pair potential, and consider
	\begin{align}\label{def-pair-potential}
		\begin{split}
			\MoveEqLeft \Vinter_{j,\alpha_j}\big((a_{\omega_{3-j}})_{\omega_{3-j}\in\Omega_{3-j}}\big)  \, = \, \frac12\sum_{\omega_{3-j}\in\Omega_{3-j}}v_{\balph}(a_{\omega_{3-j}}) \, .
		\end{split}
	\end{align}
	Such interatomic pair potentials $v_\balph$ have been constructed using DFT, e.g., in \cite{leven2016interlayer,Hod10,schmidt2015interatomic}.
	
	\par The fact that we are using the ordered index $\balph=(\alpha_1,\alpha_2)$ to label pair potentials, as opposed to, e.g., $(\alpha_j,\alpha_{3-j})$, is to indicate that pair interactions only depend on the pair of involved atoms, as opposed to the order of the pair. Whenever there is a $\alpha_j$ or $\alpha_{3-j}$ on the LHS of an equation, and $\alpha_1$ or $\alpha_2$ on the RHS, the corresponding equation is to be read individually for $j=1$ or $j=2$, to avoid confusion.

    We abbreviate the \emph{(total) finite difference stencils}
	\begin{align}
		\Delta_{j,3-j}[\bu](x,\alpha_j) \, &:= \, \big(Y_{3-j}(\omega_{3-j}+A_{3-j}A_j^{-1}x+\gamma_{3-j})-Y_j(x+\gamma_j,\alpha_j) \big)_{\omega_{3-j}\in\Omega_{3-j}} \, .\label{def-fin-diff-stenc-inter}
	\end{align} 
	With that, we define the truncated interlayer coupling energy density of the $j$-th layer	
	\begin{align} \label{def-interlayer}
		\begin{split}
			\inter_{j,N,\bgam}(\bu) \, := \, \frac{|\cA_j|}{|\Gamma_j||\Omega_j(N)|} \sum_{\omega_j\in\Omega_j(N)}\Vinter_{j,\alpha_j}\big(\Delta_{j,3-j}[\bu](\omega_j)\big) ,
		\end{split}
	\end{align}
    and its formal limit
    \begin{align}\label{def-interj}
		\inter_{j,\bgam}(\bu) \, := \, \frac1{|\Gamma_j|}\sum_{\alpha_j\in\cA_j} \mavint \dx{x} \Vinter_{j,\alpha_j}\big(\Delta_{j,3-j}[\bu](x,\alpha_j)\big) \, .
	\end{align}
    Again, we define the site potential density
    \begin{align}
		\sitepotinter_{j,\bgam}[\bu]: \R^2\to\R, \, x\mapsto & \frac1{|\Gamma_j|}\sum_{\alpha_j\in\cA_j} \Vinter_{j,\alpha_j}\big(\Delta_{j,3-j}[\bu](x,\alpha_j)\big) \, , \label{def-sitepotinter}
	\end{align}    
    and the offset site potential
    \begin{align}
		\MoveEqLeft\dsitepotinter_{j,\bgam}[\bdu;\bu] \, := \, |\Gamma_j|\big(\sitepotinter_{j,\bgam}[\bu+\bdu]-\sitepotinter_{j,\bgam}[\bu]\big)\\
        &= \, \sum_{\alpha_j\in\cA_j} \Big(\Vinter_{j,\alpha_j}\big(\Delta_{j,3-j}[\bu+\bdu](\cdot,\alpha_j)\big)-\Vinter_{j,\alpha_j}\big(\Delta_{j,3-j}[\bu](\cdot,\alpha_j)\big)\Big) \, .\label{def-sitepotinter-offset}
	\end{align}
    This allows us to define the interlayer offset energy
    \begin{align}
		\deinter_{j,\bgam;\bu}(\bdu) \, &:= \, \sum_{R_j\in\cR_j}\dsitepotinter_{j,\bgam}[\bdu;\bu](R_j) \, .\label{def-deinterj}
	\end{align}

    \subsubsection{Total energy}
    We abbreviate the truncated total energy density
    \begin{align}\label{def-etot-trunc}
        \etot_{N,\bgam}(\bu) \, := \, \sum_{j=1}^2\big(\mono_{j,N,\gamma_j}(u_j)+\inter_{j,N,\bgam}(\bu)\big) \, ,
    \end{align}
    the limit energy
    \begin{align}
		\etot_\bgam(\bu) \, := \, \sum_{j=1}^2\big(\mono_j(u_j)+\inter_{j,\bgam}(\bu)\big) \, .
	\end{align}
    as well as the total offset energy
    \begin{align}
        \detot_{\bgam;\bu}(\bdu) \, := \, \sum_{j=1}^2\big(\demono_{j;u_j}(\du_j)+\deinter_{j,\bgam;\bu}(\bdu)\big) \, .
    \end{align}
    Below, we will show the convergence of $\detot_{\bu}(\bdu)$ and of
    \begin{align}
        \lim_{N\to\infty}\etot_{N,\bgam}(\bu) \, = \, \etot_\bgam (\bu) 
    \end{align}
    for displacements $\bu$ and $\bdu$ in appropriate function spaces, and sufficiently regular potentials $\Vmono$, $\Vinter$.
    \par Below, if the explicit $\bgam$-dependence is not important, we will also omit it in the notation of $\etot$, $\detot$, etc. 
    
 \subsection{Assumptions and function spaces\label{sec-assumptions}}

    We largely follow \cite{ortnertheil2013} and adapt the assumptions to the present context. 
    
    \subsubsection{Displacements} 
    We start by defining the appropriate function spaces for the energy minimizers. Observe that $\mono_j(u_j)$ is invariant under the shift $u_j$, and that shifting $\bu\to\bu+(t_1,t_2)$ with $t_j\in\R^2$ in $\inter_j(\bu)$ amounts to redefining the origins $\gamma_j\to \Lfrac{\gamma_j+t_j}{3-j}$. Consequently, we assume that
	\begin{align}
		\mavint\dx{x} u_j(x) \, = \, 0 \, .
	\end{align}
	Let
	\begin{align}\label{def-Wbper}
		\MoveEqLeft\Wbper^{1,2}:=\Wbper^{1,2}(\MoireCell;(\R^2)^{\cA_1}\times(\R^2)^{\cA_2}) \, := \\
		&\quad \Big\{\bu\big|_{\MoireCell}\mid \bu\in \Lloc^2(\R^2;(\R^2)^{\cA_1}\times(\R^2)^{\cA_2}),\\
		&\qquad \qquad \quad \bu(x)=\bu(x+\MoireR) \quad \forall \MoireR\in\MoireRL, \, \mathrm{a.e.\ }x\in\R^2,\\
		&\qquad \qquad  \quad \bu\big|_{\MoireCell}\in W^{1,2}(\MoireCell;(\R^2)^{\cA_1}\times(\R^2)^{\cA_2}), \mavint\dx{x}\bu(x) \,= \, 0 \Big\} \, .
	\end{align}
    We endow $\Wbper^{k,p}$ with the $\dot{W}^{k,p}$-topology, i.e.,
    \begin{equation}
        \|\bu\|_{\Wbper^{k,p}} \, := \, \|\nabla^k \bu\|_p \, .
    \end{equation}
	Analogously, we define $\Wbper^{1,2}(\MoireCell;(\R^2)^{\cA_j})$, $j=1,2$, as well as $\Wbper^{n,p}$, etc. Next, let
	\begin{align}\label{def-min-space}
		\cD_\kappa \, :=& \, \Big\{\bu=(u_1,u_2)\in\Wbper^{1,2}(\MoireCell;(\R^2)^{\cA_1}\times(\R^2)^{\cA_2})\mid \forall j=1,2:\|\nabla u_j\|_\infty\leq \kappa,\\
		&\quad \sup_{R_j\in\cR_j}\|u_j(\cdot+R_j)-u_j\|_\infty,\sup_{R_j\in\cR_j}\|u_j(A_jA_{3-j}^{-1}\cdot+R_j)-u_{3-j}\|_\infty\leq\kappa \Big\} \, .
	\end{align}

    \subsubsection{Interpolants} 
    
    When discussing the Cauchy-Born approximation, we will restrict ourselves to a single sublattice degree of freedom per layer to simplify the analysis. This will allow us to immediately adapt ideas from  \cite{ortnertheil2013}. To extent the ideas to multilattices, we refer to \cite{ortner2018CB-multilattice}. 
    \par Let $\zeta\in C^\infty_c(\R^2)$ be such that $\supp \zeta\cap \Z^2=\{0\}$, $\zeta(0)=1$, $\zeta\geq0$, $\zeta(x)=\zeta(-x)$, and satisfying 
		\begin{align}
			\sum_{n\in\Z^2}\zeta(x-n) (a+b\cdot n) \, = \, a+b\cdot x \quad \forall{a\in\R, \, b,x\in\R^2} \, .
		\end{align}
		In particular, we have that $\int_{\R^2}\dx{x}\zeta(x)=1$. We also abbreviate $\zeta_j:=\zeta(A_j^{-1}\cdot)$.
        \par Given a function $v_j:\cR_j\to \R$, we construct the \emph{first order interpolant} 
	\begin{equation}\label{def-fo-interpol}
		v_j(x) \, := \, \sum_{R_j\in\cR_j} \zeta_j(x-R_j)v_j(R_j) \, ,
	\end{equation}
    as well as the quasi-interpolant
    \begin{equation}\label{def-quasi-interpolant}
        \widetilde{v}_j \, := \, \frac1{|\Gamma_j|}\zeta_j*v_j \, .
    \end{equation}
    A useful property of convolving with $\zeta_j$ is that finite differences can be rewritten as convolution with a directional gradient 
        \begin{align}
            \tilde{v}_j (x+R_j)-\tilde{v}_j (x) \, &= \, \int_0^1\dx{t} \nabla_{R_j}\tilde{v}_j(x+tR_j) \\
            &= \, \frac1{|\Gamma_j|}\int_{\R^2}\dx{y}\int_0^1\dx{t}\zeta_j(x+tR_j-y)\nabla_{R_j}v(y) \, . \label{eq-localization}
        \end{align}
    
    \subsubsection{Potentials}

	\paragraph{Smoothness and decay}
	We define the spaces for finite difference stencils
	\begin{align}
		\frakD_j^\kappa  := \Big\{\ba=(a_{\omega_j})_{\omega_j\in\Omega_j}\mid a_{\omega_j}\in \R^2, \, \|\ba\|_{\frakD_j}:=\sup_{\substack{\omega_j\in\Omega_j,\\\omega_j\neq(0,\alpha_j)}}\frac{|a_{\omega_j}|}{|R_j|}\leq \kappa\Big\} \, .
	\end{align}
	
    Let $\bV^m=(\bv^{(1)},\ldots,\bv^{(m)})$ and $\bv^{(\ell)}=(v^{(\ell)}_1,v^{(\ell)}_2)$, where $v^{(\ell)}_j\in \Wbper(\MoireCell;(\R^2)^{\cA_j})$. We abbreviate the interlayer difference stencils
	\begin{align}
		(\delji_{\omega_{3-j}}\bv^{(\ell)})(x,\alpha_j) \, &:= \, v^{(\ell)}_{3-j}(\omega_{3-j}+\gamma_{3-j}+A_{3-j}A_j^{-1}x)-v^{(\ell)}_j(x+\gamma_j,\alpha_j) \, , \\
		(\delji_{\bomeg_{3-j}^m}\bV^m)(x,\alpha_j) \, &:= \, \big((\delji_{\omega_{3-j}^{(\ell)}}\bv^{(\ell)})(x,\alpha_j)\big)_{\ell=1}^m \, .
	\end{align}
	Similarly, denoting $\bV^m_j:=(v^{(1)}_j,\ldots,v^{(m)}_j)$, we introduce the monolayer difference stencils
	\begin{align}
		(\delj_{\omega_j'}v^{(\ell)}_j)(x,\alpha_j) \, &:= \, v^{(\ell)}_j(\omega_j'+x)-v^{(\ell)}_j(x,\alpha_j) \, ,\\
		(\delj_{\bomeg_j^m}\bV^m_j)(x,\alpha_j) \, &:= \, \big((\delj_{\omega_j^{(\ell)}}v^{(\ell)}_j)(x,\alpha_j)\big)_{\ell=1}^m \, .
	\end{align}
	In addition, we abbreviate the partial derivatives 
	\begin{align}
		\Vmono_{j,\alpha_j;\bomeg_j^m}(\ba) \, &:= \, (D_{\bomeg_j^m}\Vmono_{j,\alpha_j})(\ba) \, := \,  \frac{\partial^m\Vmono_{j,\alpha_j}(\ba)}{\partial a_{\omega_j^{(1)}}\ldots \partial a_{\omega_j^{(m)}}} \, ,\\
		\Vinter_{j,\alpha_j;\bomeg_j^m}(\ba) \, &:= \, (D_{\bomeg_j^m}\Vinter_{j,\alpha_j})(\ba) \, := \,  \frac{\partial^m\Vinter_{j,\alpha_j}(\ba)}{\partial a_{\omega_j^{(1)}}\ldots \partial a_{\omega_j^{(m)}}} \, .
	\end{align}
    With that, we introduce the intralayer site potential norm
	\begin{align}
		\mmono_{\ell,j}(\bomeg_j^{\ell}) \,&:= \, \Big(\prod_{n=1}^\ell |R_j^{(n)}|\Big)\sup_{\substack{\bh^{(\ell)}\in(\R^2)^\ell,|h_n^\ell|=1,\\\ba\in\frakD_j^\kappa}}\sum_{\alpha_j\in\cA_j}|\Vmono_{j,\alpha_j;\bomeg_j^{\ell}}(\ba)[\bh^{(\ell)}]| \, ,\\
		\minter_{\ell,3-j}(\bomeg_j^{\ell}) \,&:= \, 
		\sup_{\substack{\bh^{(\ell)}\in(\R^2)^\ell,|h_n^\ell|=1,\\\ba\in\frakD_j^\kappa,x\in\MoireCell}}	\sum_{\substack{\alpha_{3-j}\\\in\cA_{3-j}}} |\Vinter_{3-j,\alpha_{3-j};\bomeg_j^{\ell}}\big(\Delta_{3-j,j}[{\bm 0}](x,\alpha_{3-j})+\ba\big)[\bh^{(\ell)}]| \, .
	\end{align}
	
	\begin{assumption}[Decay assumption]\label{ass-decay}
		We assume that there exists $k\in\N^{\geq 5}$ such that for all $j\in\{1,2\}$ and $\ell\in\{1,\ldots, k\}$, we have that
		\begin{align}
			\Mmono_{\ell,j} := \sum_{\bomeg_j^{\ell}\in\Omega_j^{\ell}} \mmono_{\ell,j}(\bomeg_j^{\ell}) \, < \, \infty \, , \quad
			\Minter_{\ell,3-j} := \sum_{\bomeg_j^{\ell}\in\Omega_j^{\ell}} \minter_{\ell,3-j}(\bomeg_j^{\ell}) \, < \, \infty \, .
		\end{align}
		We also abbreviate $M_\ell^{(m)}:=\sum_{j=1}^2M_{\ell,j}^{(m)}$ for $m\in\{\mathrm{mono},\mathrm{inter}\}$. In addition, we assume that
		\begin{equation}
			\Vmono_{j,\alpha_j}({\bf 0})\, = \, 0 , \quad \sup_{x\in\MoireCell}\sup_{\ba\in\frakD_j^\kappa}|\Vinter_{3-j,\alpha_{3-j}}\big(\Delta_{3-j,j}[{\bm 0}](x,\alpha_{3-j})+\ba\big)|<\infty \, .
		\end{equation}
	\end{assumption}

    As we will see, this decay assumption suffices to establish the well-posedness of the relaxation problem in the \emph{Cauchy-Born approximation}. In order to ensure the well-posedness for the relaxation problem associated with the general atomistic energy density, we will impose the following condition.

    \par For $p\in[1,\infty]$, $2\leq \ell \leq m$, and $\bomeg_j^\ell\in\Omega_j^\ell$, where $\omega_j^{(n)}=(R_j^{(n)},\alpha_j^{(n)})\in\Omega_j$, we define 
	\begin{align}
		\MoveEqLeft\mmonoat_{\ell,p,j}(\bomeg_j^\ell) \, := \\
		& \mmono_{\ell,j}(\bomeg_j^\ell)\Big(\sum_{n=1}^\ell|R_j^{(n)}|\Big)^2\Big(\sum_{n=2}^\ell\big(|R_j^{(1)}\times R_j^{(n)}|+|R_j^{(1)}|+|R_j^{(n)}|\big)\Big)^{\frac1{(\ell-1)p}} \, .
	\end{align}
	
	\begin{assumption}[Sufficient decay for modelling error]
		We assume that there exists $k\in\N^{\geq 3}$ such that for all $j\in\{1,2\}$ and $\ell\in\{1,\ldots, k\}$, we have that
		\begin{align}
			\Mmonoat_{\ell,p,j} \, := \, \sum_{\bomeg_j^\ell\in\Omega_j^\ell}\mmonoat_{\ell,p,j}(\bomeg_j^\ell) \, < \, \infty \, ,
		\end{align}
		and we abbreviate $\Mmonoat_{\ell,p} :=\sum_{j=1}^2\Mmonoat_{\ell,p,j}$.
	\end{assumption}

    \paragraph{Stability} As part of our assumptions, we will impose a lattice stability condition. This condition suffices to show the existence of a locally unique minimizer of $\etot$.

    \begin{assumption}[Stacking-based stability of lattice]
		We assume that
		\begin{align}
			\nu_{\mathrm{Stack}} \, := \, \min_{j=1,2} \inf_{\|\nabla v\|_2=1} \delta^2\mono_j(0)[v,v] \, > \, 0 \, .
		\end{align}
	\end{assumption} 

    One could state an analogous assumption in terms $\delta^2\demono_j(0)$. Crucially, such a condition in addition to smallness of $\delta^2\deinter_j(0)$ does not suffice to prove existence of a locally unique minimizer, nor its stability.
 
	\subsection{Convergence of energy functionals}

    We now address the convergence of the total energy density $\etot(\bu)$ as well as of the total offset energy $\detot_{\bu}(\bdu)$. Crucially, we will see that $\bu$ will be uniformly non-decaying, think $\Lloc^1$, while $\bdu$ will be required to decay in $\ell^1(\Omega)$.
    
    \subsubsection{Energy density}

    We start by recalling some results established in \cite{hott2023incommensurate}. A useful tool to prove the existence of the thermodynamic limit $\etot$ is the following \emph{Birkhoff ergodic theorem}, see \cite{KerrLi-erg-th} and also \cite[Proposition 3.3]{hott2023incommensurate} for this formulation. 
	
	\begin{proposition}[Ergodic Theorem]\label{prop-double-erg-thm}
		Let $f\in \Lper^1(\MoireCell;\Lper^1(\Gamma_j))$. Then
		\begin{align}
			\begin{split}
				\MoveEqLeft \lim_{N\to\infty}\sum_{\substack{R_j\in\\\cR_j(N)}}\frac{f(R_j+\omega_\cM,R_j+\omega_{3-j})}{(2N+1)^2} = \mavint \dx{x} f\big(x+\omega_\cM,(I-A_{3-j}A_j^{-1})x+\omega_{3-j}\big) 
			\end{split}
		\end{align}
		converges for almost all $(\omega_\cM,\omega_{3-j})\in\MoireCell\times\Gamma_{3-j}$ and also in $L^1(\MoireCell\times\Gamma_{3-j})$.
	   \end{proposition}
    
    In order to employ Proposition \ref{prop-double-erg-thm} to $(\etot_{N,\bgam})_N$, we also introduce the following generalized site potential
    \begin{align}\label{def-sitepotdouble}
\begin{split}
	\MoveEqLeft\sitepotdouble_j [\bu](x,y) \, :=\\
	&\frac1{|\Gamma_j|}\sum_{\alpha_j\in\cA_j} \Vinter_{j,\alpha_j}\Big(\big(Y_{3-j}(x-y+R_{3-j},\alpha_{3-j})-Y_j(x,\alpha_j)\big)_{\omega_{3-j}\in\Omega_{3-j}}\Big) \, .
    \end{split}
    \end{align}

    \begin{theorem}\label{thm-rough-conv}
	    Let $\bu\in\cD_{\kappa}$. Assume that $\sitepotmono_j [u_j]\in L^1(\MoireCell)$, and that $
		\sitepotdouble_j [\bu]\in L^1(\MoireCell\times\Gamma_{3-j})$ and for any $\bgam=(\gamma_1,\gamma_2)\in\Gamma_1\times\Gamma_2$ that $\sitepotinter_{j,\bgam} [\bu]\in L^1(\MoireCell)$, respectively. 
        We then have for almost all $\bgam\in\Gamma_1\times \Gamma_2$ that 
	\begin{align} 
        \lim_{N\to\infty}\etot_{N,\bgam}(\bu) \, = \, \etot_\bgam(\bu) \, .
	   \end{align}.
    \end{theorem}

    \begin{remark}
        In \cite{hott2023incommensurate}, it was also proved that, under a Diophantine condition on $A_1$, $A_2$, and requiring additional regularity of $\bu$, $\Vmono$, and $\Vinter$, one obtains a rate of convergence of $O(N^{-1})$. There, the result was only proved in the case $\bgam={\bf 0}$, but the ideas presented there can be extended to the case $\bgam\neq{\bf 0}$. The Diophantine condition presented in \cite{hott2023incommensurate} can be, more generally, stated as
        \begin{equation}
			\dist(A_1^TA_2^{-T}n,\Z^2)\geq \frac{K}{|n|^{2\sigma}} \quad \forall n\in\Z^2\setminus\{0\}\, .
		\end{equation}
        In particular, Diophantineness is a quantification of incommensurability.
    \end{remark}

    We are left with proving that the assumptions on $\Vmono$ and $\Vinter$ stated in Section \ref{sec-assumptions} imply the conditions of Theorem \ref{thm-rough-conv}.

    \par Our assumptions immediately imply that $\sitepotinter_j\in L^1(\MoireCell)$ and that $\sitepotdouble_j\in L^1(\MoireCell\times\Gamma_{3-j})$. In the case of $\sitepotmono_j$, observe that
    \begin{align}
        |\sitepotmono_j[u_j](x)-\underbrace{\sitepotmono_j[0]}_{=0}| \, &\lesssim \, \frac1{|\Gamma_j|}\sup_{\alpha_j}\sum_{\omega_j'\in\Omega_j'}\mmono_{1,j}(\omega_j')\frac{|u_j(\omega_j'+x)-u_j(x,\alpha_j)|}{|R_j'|}\\
        &\lesssim \, \frac{\kappa\Mmono_{1,j}}{|\Gamma_j|} \, ,
    \end{align}
	where we used the fact that $\bu\in\cD_\kappa$. In particular, we obtain that $\sitepotmono_j\in L^1(\MoireCell)$.

    \subsubsection{Offset energy}

    Expanding $\detot_{\bu}$ up to second order in $\bdu$, we obtain
	\begin{align}
		\MoveEqLeft\detot_{\bu}(\bdu)\, = \\ 
		& \sum_{j=1}^2\sum_{\omega_j\in \Omega_j}\Big[\sum_{\omega_j'}\Vmono_{j,\alpha_j;\omega_j'}\big(\Delta_j[\ueq_j](\omega_j)\big)\cdot \delj_{\omega_j'}\du_j(\omega_j)\\
		& \quad +\int_0^1\dx{t}(1-t)\hspace{-1ex}\sum_{\omega_j',\omega_j''}\hspace{-1ex}\Vmono_{j,\alpha_j;\omega_j',\omega_j''}\big(\Delta_j[\ueq_j+t\du_j](\omega_j)\big)\\
		&\qquad \big[\delj_{\omega_j'}\du_j(\omega_j)),\delj_{\omega_j''}\du_j(\omega_j)\big]  \\
		&\quad +\sum_{\omega_{3-j}}\Vinter_{j,\alpha_j;\omega_{3-j}}\big(\Delta_{j,3-j}[\bueq](\omega_j)\big)\cdot \delji_{\omega_{3-j}}\bdu(\omega_j)\\
		&\quad +\int_0^1\dx{t}(1-t) \hspace{-3ex}\sum_{\omega_{3-j},\omega_{3-j}'}\Vinter_{j,\alpha_j;\omega_{3-j},\omega_{3-j}'}\big(\Delta_{j,3-j}[\bueq+t\bdu](\omega_j)\big)\\
		& \qquad \big[\delji_{\omega_{3-j}}\bdu(\omega_j), \delji_{\omega_{3-j}'}\bdu(\omega_j)\big]\Big] \, .
	\end{align}
    A straight-forward calculation then yields that
    \begin{align}
        |\detot_{\bu}(\bdu)| \, \lesssim \, (\Mmono_1+\Minter_1)\|\bdu\|_{\ell^1} \, + \, (\Mmono_2+\Minter_2)\|\bdu\|_{\ell^2}^2 \, .
    \end{align}
        
	\subsection{First variation\label{sec-first-var}}
	
	Our goal is to derive the expressions for the first Frech\'et derivative of the energy density as well as of the offset energy. For now, we will assume sufficient regularity of $\etot$ and $\detot_{\bu}$, and assume the existence of a critical point $\bueq=(\ueq_1,\ueq_2)$ of $\etot$. In Section \ref{sec-energy-well-def}, we will show that our assumptions on $\Vmono$ and $\Vinter$ suffice to justify the following results.

	For simplicity, let us set $\gamma_1=\gamma_2=0$. It is useful to introduce the following notation: 
	
	\subsubsection{Thermodynamic limit energy density} For $\bv=(v_1,v_2)\in \Wbper$, we have that
	\begin{align}
		& \delta\etot(\bu)[\bv] \, = \\
		&\quad \sum_{j=1}^2\frac1{|\Gamma_j|}\sum_{\alpha_j\in\cA_j}\mavint \dx{x} \Big(\sum_{\omega_j'}\Vmono_{j,\alpha_j;\omega_j'}\big(\Delta_j[u_j](x,\alpha_j)\big)\cdot \big(v_j(\omega_j'+x)-v_j(x,\alpha_j)\big)\\
		&\quad + \sum_{\substack{\omega_{3-j}}}\Vinter_{j,\alpha_j;\omega_{3-j}}\big(\Delta_{j,3-j}[\bu](x,\alpha_j)\big)\cdot \big(v_{3-j}(\omega_{3-j}+A_{3-j}A_j^{-1}x)-v_j(x,\alpha_j)\big)\Big) \, .
	\end{align}

	We now sort terms according to the coefficients $v_j$
	\begin{align}
		&\delta\etot(\bu)[\bv] \, = \\
		&\quad \sum_{j=1}^2\frac1{|\Gamma_j||\MoireCell|}\sum_{\alpha_j\in\cA_j}\int\dx{x}\Big[ \sum_{\omega_j'}\Big(\Vmono_{j,\alpha_j;\omega_j'}\big(\Delta_j[u_j](x-R_j',\alpha_j')\big)\mathds{1}_{\MoireCell+R_j'}(x)\\
		&\quad -\Vmono_{j,\alpha_j;\omega_j'}\big(\Delta_j[u_j](x,\alpha_j)\big)\mathds{1}_{\MoireCell}(x)\Big)- \sum_{\omega_{3-j}}\Vinter_{j,\alpha_j;\omega_{3-j}} \big(\Delta_{j,3-j}[\bu](x,\alpha_j)\big)\mathds{1}_{\MoireCell}(x)\\
		&\quad  + \sum_{R_j,\alpha_{3-j}}\Vinter_{3-j,\alpha_{3-j};\omega_j} \big(\Delta_{3-j,j}[\bu](A_{3-j}A_j^{-1}(x-R_j),\alpha_{3-j})\big)\mathds{1}_{A_jA_{3-j}^{-1}\MoireCell+R_j}(x) \Big]\\
		&\quad  \cdot v_j(x,\alpha_j) \, .
	\end{align}
	Observe that $x\mapsto D_{R_j}\Vmono(\Delta_j[u_j](x))$ and $v_j$ are $\MoireRL$-periodic. That allows us to simplify the expression as
	\begin{align}
		&\delta\etot(\bu)[\bv] \, = \\
		&\quad \sum_{j=1}^2\frac1{|\Gamma_j|}\sum_{\alpha_j\in\cA_j}\mavint\dx{x}\Big[ \sum_{\omega_j'}\Big(\Vmono_{j,\alpha_j;\omega_j'}\big(\Delta_j[u_j](x-R_j',\alpha_j')\big)\\
		&\quad -\Vmono_{j,\alpha_j;\omega_j'}\big(\Delta_j[u_j](x,\alpha_j)\big)\Big)\\
		&\quad + \sum_{R_j,\MoireR,\alpha_{3-j}}\Vinter_{3-j,\alpha_{3-j};\omega_j} \big(\Delta_{3-j,j}[\bu](A_{3-j}A_j^{-1}(x-R_j-\MoireR),\alpha_{3-j})\big)\\
		&\quad \mathds{1}_{A_jA_{3-j}^{-1}\MoireCell+R_j+\MoireR}(x) - \sum_{\omega_{3-j}}\Vinter_{j,\alpha_j;\omega_{3-j}}\big(\Delta_{j,3-j}[\bu](x,\alpha_j)\big)\Big] \cdot v_j(x,\alpha_j) \, .
	\end{align}
	We would like to simplify the the third term. For that, observe that translation-invariance of $\Vinter$ implies
	\begin{align}\label{eq-DVinter-TI}
		\Vinter_{3-j;R_j}\big((a_{R_j'})_{R_j'}\big) \, = \, \Vinter_{3-j;0}\big((a_{R_j'+R_j})_{R_j'}\big) \, .
	\end{align}
	Ignoring the $\balph$-dependence for now, we thus find that
	\begin{align}
		&\sum_{R_j,\MoireR}\Vinter_{3-j;R_j}\big(\Delta_{3-j,j}[\bu](A_{3-j}A_j^{-1}(x-R_j-\MoireR))\big)\mathds{1}_{A_jA_{3-j}^{-1}\MoireCell+R_j+\MoireR}(x) \\
		&= \, \sum_{R_j,\MoireR}\Vinter_{3-j;R_j}\Big(\big(Y_{3-j}(A_{3-j}A_j^{-1}(x-R_j-\MoireR))\\
		& \qquad -Y_j(x-R_j-\MoireR+R_j')\big)_{R_j'}\Big)\mathds{1}_{A_jA_{3-j}^{-1}\MoireCell+R_j+\MoireR}(x)\\
		&= \, \sum_{R_j,\MoireR}\Vinter_{3-j;0}\Big(\big(Y_{3-j}(A_{3-j}A_j^{-1}(x-R_j-\MoireR))\\
		& \qquad -Y_j(x-\MoireR+R_j')\big)_{R_j'}\Big)\mathds{1}_{A_jA_{3-j}^{-1}\MoireCell+R_j+\MoireR}(x) \, . \label{eq-Vinter-massage-0}
	\end{align}
	Thanks to Lemma \ref{lem-disregj-calc}, we have that
	\begin{align*}
		-A_{3-j}A_j^{-1}\MoireR +\MoireR \, = \, \disregj \MoireR \, = \, (-1)^jA_{3-j}\MoireRLV^{-1}\MoireR \, .
	\end{align*}
	Thus we obtain
	\begin{align}\label{eq-moireR-shift}
		-A_{3-j}A_j^{-1}(R_j+\MoireR) +\MoireR \, = \, (-1)^{j+1}A_{3-j}\MoireRLV^{-1}\big((-1)^j\MoireRLV A_j^{-1}R_j-\MoireR \big) \, .
	\end{align}
	Due to $\MoireRL$-periodicity of $u_{3-j}$, we thus find that
	\begin{align}
		&u_{3-j}\big(A_{3-j}A_j^{-1}(x-R_j-\MoireR)\big) \, = \, u_{3-j}\big(A_{3-j}A_j^{-1}(x-R_j-\MoireR)+\MoireR\big) \\
		&\quad = \, u_{3-j}\big(A_{3-j}A_j^{-1}x+(-1)^{j+1}A_{3-j}\MoireRLV^{-1}\big((-1)^j\MoireRLV A_j^{-1}R_j-\MoireR \big)\big) \, . \label{eq-moireR-u-shift}
	\end{align}
	As a consequence of \eqref{eq-moireR-shift} and \eqref{eq-moireR-u-shift}, we shift $\MoireR\to \MoireR+(-1)^j\MoireRLV A_j^{-1}R_j$ in \eqref{eq-Vinter-massage-0} to obtain
	\begin{align}
		&\sum_{\MoireR,R_j}\Vinter_{3-j;0}\Big(\big(Y_{3-j}(A_{3-j}A_j^{-1}(x-\MoireR))-Y_j(x-\MoireR+R_j')\big)_{R_j'}\Big)\\
		& \qquad \mathds{1}_{A_jA_{3-j}^{-1}\MoireCell+(I+(-1)^j\MoireRLV A_j^{-1})R_j+\MoireR}(x) \, . \label{eq-DVinter-simplify-1}
	\end{align}
	Notice the only $R_j$-dependence is on the characteristic function. Moreover, we have that
	\begin{align*}
		(I+(-1)^j\MoireRLV A_j^{-1})R_j \, &= \, A_j(\MoireRLV^{-1}+(-1)^jA_j^{-1})\MoireRLV A_j^{-1}R_j \\
		&= \, (-1)^{j+1} A_jA_{3-j}^{-1} \MoireRLV A_j^{-1}R_j \, ,
	\end{align*}
	which, in turn, implies
	\begin{align}\label{eq-POU-1}
		\sum_{R_j}\mathds{1}_{A_jA_{3-j}^{-1}\MoireCell+(I+(-1)^j\MoireRLV A_j^{-1})R_j+\MoireR}(x) \, = \, 1 \, .
	\end{align}
	Combining \eqref{eq-DVinter-simplify-1} and \eqref{eq-POU-1}, we obtain
	\begin{align}\label{eq-DVinter-simplify-2}
		\sum_{\MoireR}\Vinter_{3-j;0}\Big(\big(Y_{3-j}(A_{3-j}A_j^{-1}(x-\MoireR))-Y_j(x-\MoireR+R_j')\big)_{R_j'}\Big) \, .
	\end{align}
	Using again $\MoireRL$-periodicity of $u_k$, we compute
	\begin{align*}
		&Y_{3-j}(A_{3-j}A_j^{-1}(x-\MoireR))-Y_j(x-\MoireR+R_j') \, = \, \\
		&\quad (A_{3-j}A_j^{-1}-I)x+(I-A_{3-j}A_j^{-1})\MoireR-R_j' \\
		&\qquad + \, u_{3-j}(A_{3-j}A_j^{-1}x+(I-A_{3-j}A_j^{-1})\MoireR) \, - \, u_j(x+R_j') \, .
	\end{align*}
	Abbreviating $R_j:=-A_jA_{3-j}^{-1}(I-A_{3-j}A_j^{-1})\MoireR\in\cR_j$, due to Lemma \ref{lem-disregj-calc}, we obtain
	\begin{align}
		\begin{split}
			&Y_{3-j}(A_{3-j}A_j^{-1}(x-\MoireR))-Y_j(x-\MoireR+R_j')  \\
			&\quad = \, (A_{3-j}A_j^{-1}-I)(x-R_j)-(R_j'+R_j)\\
			&\qquad\quad  + \, u_{3-j}(A_{3-j}A_j^{-1}(x-R_j)) \, - \, u_j(x-R_j+(R_j'+R_j)) \\
			&\quad = \, Y_{3-j}(A_{3-j}A_j^{-1}(x-R_j)) \, - \, Y_j(x-R_j+(R_j'+R_j)) \, . 
		\end{split}\label{eq-inter-diff-vec-simplify-1}
	\end{align}
	Employing \eqref{eq-DVinter-TI}, \eqref{eq-DVinter-simplify-2} and \eqref{eq-inter-diff-vec-simplify-1} yield 
	\begin{align}
		\sum_{R_j}\Vinter_{3-j;R_j}\Big(\Delta_{3-j,j}[\bu]\big(A_{3-j}A_j^{-1}(x-R_j)\big)\Big) \, .
	\end{align}
	Altogether, we have thus proved that
	\begin{align}
		&\delta\etot(\bu)[\bv] \, = \\
		&\quad \sum_{j=1}^2\frac1{|\Gamma_j|}\sum_{\alpha_j\in\cA_j}\mavint\dx{x}\Big[ \sum_{\omega_j'}\Big(\Vmono_{j,\alpha_j;\omega_j'}\big(\Delta_j[u_j](x-R_j',\alpha_j')\big)\\
		&\quad -\Vmono_{j,\alpha_j;\omega_j'}\big(\Delta_j[u_j](x,\alpha_j)\big)\Big)\\
		& \quad + \sum_{R_j,\alpha_{3-j}}\Vinter_{3-j,\alpha_{3-j};\omega_j}\big(\Delta_{3-j,j}[\bu]\big(A_{3-j}A_j^{-1}(x-R_j),\alpha_{3-j}\big)\big)\\
		&\quad  - \sum_{\omega_{3-j}}\Vinter_{j,\alpha_j;\omega_{3-j}}\big(\Delta_{j,3-j}[\bu](x,\alpha_j)\big)\Big] \cdot v_j(x,\alpha_j) \, .
	\end{align}
	Consequently, the Euler-Lagrange equation for the minimizer $\bueq$ is given by 
	\begin{align}\label{eq-ELE-MB-cont}
		\begin{split}
			& \sum_{\omega_j'}\Big(\Vmono_{j,\alpha_j;\omega_j'}\big(\Delta_j[\ueq_j](x-R_j',\alpha_j')\big)-\Vmono_{j,\alpha_j;\omega_j'}\big(\Delta_j[\ueq_j](x,\alpha_j)\big)\Big)\\
			&\quad + \sum_{R_j,\alpha_{3-j}}\Vinter_{3-j,\alpha_{3-j};\omega_j}\big(\Delta_{3-j,j}[\bueq]\big(A_{3-j}A_j^{-1}(x-R_j),\alpha_{3-j}\big)\big)\\
			& \quad - \sum_{\omega_{3-j}}\Vinter_{j,\alpha_j;\omega_{3-j}}\big(\Delta_{j,3-j}[\bueq](x,\alpha_j)\big) \, = \, 0 \, .
		\end{split}
	\end{align}
	
	\subsubsection{Offset energy}
	
	In order to study phonons, we now also look at the energy difference resulting from subtracting the groundstate total energy, as opposed to its energy density. 
	
    Then we define
	\begin{align}
		\MoveEqLeft\detot_{\bgam,N;\bueq}(\bdu)\, := \\ 
		& \sum_{j=1}^2\sum_{\omega_j\in \Omega_j(N)}\Big[\Big(\Vmono_{j,\alpha_j}\big(\Delta_j[\ueq_j+\du_j](\omega_j)\big)-\Vmono_{j,\alpha_j}\big(\Delta_j[\ueq_j](\omega_j)\big)\Big)\\
		& + \Big(\Vinter_{j,\alpha_j}\big(\Delta_{j,3-j}[\bueq+\bdu](\omega_j)\big)- \Vinter_{j,\alpha_j}\big(\Delta_{j,3-j}[\bueq](\omega_j)\big)\Big)\Big] \, .
	\end{align}

	For the limit $N\to\infty$ to exist, we need to impose integrability of $\bdu$. In particular, we consider $\bdu\in\ell^1(\Omega_1\times \Omega_2)$. We focus on the linear-in-$\bdu$ part and, as above, we reorder terms according to the coefficient $\du_j(\omega_j)$
	\begin{align}
		&\sum_{j=1}^2\sum_{\omega_j}\Big[\sum_{\omega_j'}\Big(\Vmono_{j,\alpha_j;\omega_j'}\big(\Delta_j[\ueq_j](R_j-R_j',\alpha_j)\big)-\Vmono_{j,\alpha_j;\omega_j'}\big(\Delta_j[\ueq_j](\omega_j)\big)\Big)\\
		&\quad + \sum_{\omega_{3-j}}\Vinter_{j,\alpha_j;\omega_{3-j}}\big(\Delta_{j,3-j}[\bueq](\omega_j)\big)\\
		& \quad -\sum_{\substack{R_j',\\\alpha_{3-j}}}\Vinter_{{3-j},\alpha_{3-j};(R_j',\alpha_j)}\big(\Delta_{3-j,j}[\bueq](A_{3-j}A_j^{-1}(R_j-R_j'),\alpha_{3-j})\big)\Big]\cdot\du_j(\omega_j) \, .
	\end{align}
	We recover the Euler-Lagrange equation \eqref{eq-ELE-MB-cont}, evaluated at the lattice positions $R_j$. In particular, the first variation $\delta\Delta_\infty E(0)=0$ vanishes, which means the solution $\bueq$ of \eqref{eq-ELE-MB-cont} is a critical point under $\ell^1$-perturbations.

    \subsection{Models\label{sec-models}}

    \subsubsection{Cauchy-Born approximation}
	
    A common further simplification to the Cauchy-Born approximation is the \emph{linear elasticity approximation}
	\begin{align}
		W_j(M) \, :=& \, M:C_j:M \, , \label{def-W}\\
		C_{j,kl mn} \, :=& \, \lambda_j \delta_{kl}\delta_{mn} \, + \, \mu_j(\delta_{kn}\delta_{lm}+\delta_{km}\delta_{ln}) \, , \label{def-C}
	\end{align}
	where $\lambda_j,\mu_j>0$ denote the \emph{Lam\'e parameters} of an individual layer. Then the monolayer energy density is given by
	\begin{align}
		\frac12\mavint\dx{x} \nabla u_j(x):C_j:\nabla u_j(x) \, .
	\end{align}

	In order to allow for out-of-plane relaxation, we need to extend the monolayer energy model beyond linear elasticity, e.g., to 
	\begin{align}
		\mavint \dx{x}\frac12\Big(\lambda_j \sum_{k=1}^2\vep_{k,k}^{(j)}(x)^2+2\mu_j \sum_{k,\ell=1}^2\vep_{k,\ell}^{(j)}(x)^2+\kappa_j |D^2u_{j,3}(x)|^2\Big) \, ,
	\end{align}
	where $\vep_{k,\ell}^{(j)}:=\frac12(\partial_k u_{j,\ell}+\partial_\ell u_{j,k}+\partial_ku_{j,3}\partial_\ell u_{j,3})$ is the strain tensor, $\lambda$, $\mu$ are the Lam\'e parameters as before, and $\kappa$ is the bending rigidity, see, e.g., \cite{peng2020strain}.

    \subsubsection{Pair potentials\label{sec-inter-models}}
    
    For the interlayer energy and in order to model the monolayer offset energy before taking a continuum limit, we may consider pair potential interactions. Common choices for such pair potentials include
	\begin{enumerate}
		\item \emph{Lennard-Jones potentials} $v_{\mathrm{LJ}}(X)=4\vep_0\big[\big(\frac{\sigma}{|X|}\big)^{12}-\big(\frac{\sigma}{|X|}\big)^6\big]$ for $X\in\R^3$, with equilibrium distance $2^{\frac16}\sigma$,
		\item \emph{Morse potentials} $v_{\mathrm{Morse}}(X)=E_0\big(e^{-2a(|X|-r_0)}-2e^{-a(|X|-r_0)}\big)$ for $X\in\R^3$, with equilibrium distance $r_0$,
		\item \emph{Kolmogorov-Crespi potentials} Introduced in \cite{KolmogorovCrespi,naik2019kolmogorov}, we define
		\begin{equation}
			v_{\mathrm{KC}}(X_j-X_k,n_j,n_k) \, := \, e^{-\lambda (r_{jk}-z_0)}\big(C+f(\rho_{jk})+f(\rho_{kj})\big)-A_0\Big(\frac{z_0}{r_{jk}}\Big)^6 
		\end{equation}
		for $X_j,X_k\in\R^3$, where $r_{jk}:=|X_j-X_k|$, $n_\ell$ denotes the normal vector to the $sp^2$-plane\footnote{In particular, this requires to solve the associated electronic problem as well, in order to determine the atomic orbitals.} in the vicinity of atom $\ell$, $\rho_{\ell m}^2:=r_{\ell m}^2-(n_\ell \cdot (X_\ell-X_m))^2$ and 
		\begin{equation*}
			f(\rho) \, := \, e^{-(\rho/\delta)^2}\sum_{n=0}^2C_{2n}(\rho/\delta)^{2n} \, .
		\end{equation*}
		For simplicity, one could choose $n_\ell=e_z$. In this case, the Kolmogorov potential simplifies to
		\begin{equation}
			v_{\mathrm{KC}}(r,\rho) \, = \, e^{-\lambda(r-z_0)}(C+2f(\rho))-A_0\Big(\frac{z_0}{r}\Big)^6 \, ,
		\end{equation}
		see, e.g., \cite{malena2023}. In case of bilayer graphene, the parameters are given by $(C,C_0,C_2,C_4)=(3.030, 15.71, 12.29, 4.933)\mbox{\ meV}$, $\delta=0.578\mbox{\ \AA}$, $\lambda = 3.629\mbox{\ \AA}^{-1}$, $A_0=10.238 \mbox{\ meV}$, $z_0 = 3.34\mbox{\ \AA}$.
		\item Combined potentials, e.g., $v(x,z)=v_{\mathrm{Morse}}(x)v_{\mathrm{LJ}}(z)$ for $x\in\R^2$, $z\in\R$. 
	\end{enumerate}
	
	\section{Well-definedness of energy functionals\label{sec-energy-well-def}}

    We now show that the previous calculations are meaningful by showing that the energy functionals are well-defined and sufficiently regular.

	\subsection{Atomistic model}
	
	The following statements can be proved analogously to \cite[Theorem 1]{ortnertheil2013}.
	
	\begin{lemma}\label{lem-mono-db}
		Let $j\in\{1,2\}$, $m\in\N$. Assume that $\bu=(u_1,u_2)\in\cK_\kappa$ and $v_\ell\in \Wbper^{1,2}(\MoireCell;(\R^2)^{\cA_j})$, $\ell=1,\ldots,m$. For any $1\leq p_1,\ldots,p_k\leq \infty$ such that $\sum_{\ell=1}^m\frac1{p_\ell}=1$, we have that
		\begin{align}
			|\delta^m\mono_j(u_j)[v_1,\ldots,v_m]| \, \lesssim \, \Mmono_{m,j} \prod_{\ell=1}^m\|\nabla v_\ell\|_{L^{p_\ell}(\MoireCell)} \, .
		\end{align}
	\end{lemma}

    For the next result, we recall the first-order interpolants defined in \eqref{def-fo-interpol}.

    \begin{lemma}
        \label{lem-mono-offset-db}
		Let $j\in\{1,2\}$, $m\in\N$. Assume that $\bu=(u_1,u_2)\in\cK_\kappa$ and that $\du,\dv_\ell\in \ell^1(\Omega_j)$, $\ell=1,\ldots,m$. For any $1\leq p_1,\ldots,p_k\leq \infty$ such that $\sum_{\ell=1}^m\frac1{p_\ell}=1$, we have that
		\begin{align}
			|\delta^m\demono_{j;u_j}(\du)[\dv_1,\ldots,\dv_m]| \, \lesssim \, \Mmono_{m,j} \prod_{\ell=1}^m\|\nabla\zeta\|_{p_\ell}\|\dv_\ell\|_{\ell^1(\cA_j;\ell^{p_\ell}(\R^2;\R^3))} \, .
		\end{align}
    \end{lemma}

    \begin{proof}
        In order to employ \cite[Lemma 6]{ortnertheil2013}, we shall map the deformations over $\cR_j$ to deformations over $\Z^2$. In particular, we define 
        \begin{equation}
            \dw_j(n) \, := \, \dv_j(A_jn) \, , \quad \dw_j(x) \, := \, \sum_{n\in\Z^2}\zeta(x-n)\dw_j(n) \, = \, \dv_j(A_jx) \, ,
        \end{equation}
        see \eqref{def-fo-interpol}. Defining
        \begin{equation}
            F^{\mathrm{(mono)}}(\dw_1,\ldots,\dw_m) \, := \, \delta^m\demono_{j;u_j}(\du)[\dv_1,\ldots,\dv_m] \, ,
        \end{equation}
        \cite[Lemma 6]{ortnertheil2013} then implies that
        \begin{align}
            |F^{\mathrm{(mono)}}(\dw_1,\ldots,\dw_m)| \, &\lesssim \, \Mmono_{m,j} \prod_{\ell=1}^m\|\nabla \dw_j\|_{\ell^1(\cA_j;L^{p_\ell}(\R^2;\R^3))}\\
            &\lesssim \, \Mmono_{m,j} \prod_{\ell=1}^m\frac{\|A_j\nabla \dv_\ell\|_{\ell^1(\cA_j;L^{p_\ell}(\R^2;\R^3))}}{|\Gamma_j|^{\frac1{p_\ell}}} \, . \label{eq-demono-diff-bd-1}
        \end{align}
        
        \par Next, we compute
        \begin{align}
            \|A_j\nabla \dv_\ell\|_{L^p}^p \, &= \, \sum_{R_j\in\cR_j}\int_{\Gamma_j}\dx{x}\big|\sum_{R_j'\in\cR_j}A_j\nabla\zeta_j(x+R_j-R_j')\dv_\ell(R_j)\big|^p \, . \label{eq-demono-diff-bd-2}
        \end{align}
        Using the fact that $\zeta_j$ has compact support, there exists a universal constant $r_0>0$ such that $\bigcup_{x\in\Gamma_j}\supp \zeta_j(x-\cdot)\subseteq B_{r_0}(0)\cap \cR_j$. With that, we obtain that
        \begin{align}
            \MoveEqLeft\sum_{R_j\in\cR_j}\int_{\Gamma_j}\dx{x}\big|\sum_{R_j'\in\cR_j}A_j\nabla\zeta_j(x+R_j-R_j')\dv_\ell(R_j)\big|^p \\
            &\lesssim \, \sum_{R_j\in\cR_j}\int_{\Gamma_j}\dx{x}\sum_{R_j'\in\cR_j\cap B_{r_0}(R_j)}|\dv_\ell(R_j')|^p \sum_{R_j'\in B_{r_0}(0)\cap \cR_j}|A_j\nabla\zeta_j(x-R_j')|^p\\
            &\lesssim \, \|\dv_\ell\|_{\ell^p}^p\sum_{R_j\in B_{r_0}(0)\cap \cR_j}\int_{\Gamma_j+R_j}\dx{x}|A_j\nabla\zeta_j(x)|^p \, ,\label{eq-demono-diff-bd-3}
        \end{align}
        where the constant is universal and independent of $p$. Recalling $\zeta_j(x)=\zeta(A_j^{-1}x)$, we further estimate the previous expression by
        \begin{align}
            |\Gamma_j|\|\dv_\ell\|_{\ell^p}^p\|\nabla\zeta\|_p^p \, , \label{eq-demono-diff-bd-4}
        \end{align}
        up to a universal constant. Collecting \eqref{eq-demono-diff-bd-1}, \eqref{eq-demono-diff-bd-2}, \eqref{eq-demono-diff-bd-3}, \eqref{eq-demono-diff-bd-4}, we conclude the proof.
    \end{proof}
 
	For the next result, we abbreviate for $\bv=(v_1,v_2)$, $v_j\in L^1(\MoireCell,(\R^2)^{\cA_j})$, $p\in[1,\infty]$
	\begin{equation}
		\|\bv\|_p \, := \, \sup_{\balph\in \cA_1\times\cA_2}\big(\|v_1(\cdot,\alpha_1)\|_p+\|v_2(\cdot,\alpha_2)\|_p\big) \, .
	\end{equation}
	
	\begin{lemma}\label{lem-inter-Cm-bd}
		Let $j\in\{1,2\}$, $m\in\N$. Assume that $\bu\in\cK_\kappa$ and $\bv^{(\ell)}=(v_1^{(\ell)},v_2^{(\ell)})\in \Wbper^{1,2}$, $\ell=1,\ldots,m$. For any $1\leq p_1,\ldots,p_m\leq \infty$ such that $\sum_{\ell=1}^m\frac1{p_\ell}=1$, we have that
		\begin{align}
			|\delta^m\inter_j(\bu)[\bv^{(1)},\ldots,\bv^{(m)}]| \, \lesssim \, \max\Big\{1,\frac{|\Gamma_j|}{|\Gamma_{3-j}|}\Big\}
			\Minter_{m,3-j} \prod_{\ell=1}^m\|\bv^{(\ell)}\|_{L^{p_\ell}(\MoireCell)} \, .
		\end{align}
	\end{lemma}
	
	\begin{proof}
		We follow the proof of \cite[Lemma 6]{ortnertheil2013}. Then we find that
		\begin{align}
			\MoveEqLeft \Vinter_{j,\alpha_j;\bomeg_{3-j}^m}(\Delta_{j,3-j}[\bu](x,\alpha_{3-j}))[(\Delta_{\bomeg_{3-j}^m}\bV^m)(x,\alpha_j)] \, \lesssim \\
			& \minter_{m,3-j}(\bomeg_j^\ell)\prod_{\ell=1}^m|(\Delta_{\omega_{3-j}^{(\ell)}}\bv^{(\ell)})(x,\alpha_j)| \, .
		\end{align}
		Consequently, and using $\MoireRL$-periodicity of $v_j$ we obtain that
		\begin{align}
			|\delta^m\inter_j(\bu)[\bV^m]| \, &\lesssim \, \sum_{\bomeg_{3-j}^\ell\in\Omega_{3-j}^\ell}\minter_{m,3-j}(\bomeg_j^\ell)\sup_{\alpha_j}\mavint\dx{x}\prod_{\ell=1}^m|(\Delta_{\omega_{3-j}^{(\ell)}}\bv^{(\ell)})(x,\alpha_j)|\\
			& \lesssim \, \Minter_{m,3-j}  \sup_{\substack{\bomeg_{3-j}^m,\\\alpha_j}}\prod_{\ell=1}^m(\|v_{3-j}^{(\ell)}(\omega_{3-j}^\ell+A_{3-j}A_j^{-1}\cdot)\|_{p_\ell}+\|v_j^{(\ell)}(\cdot,\alpha_j)\|_{p_\ell})\, .
		\end{align}
		Substituting $A_{3-j}A_j^{-1}x\to x$, we obtain due to $\MoireRL$-periodicity of $v_{3-j}^{(\ell)}$
		\begin{align}
			|\delta^m\inter_j(\bu)[\bV^m]| \, \lesssim \, \max\Big\{1,\frac{|\Gamma_j|}{|\Gamma_{3-j}|}\Big\}
            \Minter_{m,3-j} \prod_{\ell=1}^m\|\bv^{(\ell)}\|_{p_\ell} \, ,
		\end{align}
		concluding the proof.
	\end{proof}
	
	With analogous steps, we obtain the following result

 \begin{lemma}\label{lem-inter-offset-Cm-bd}
		Let $j\in\{1,2\}$, $m\in\N$. Assume that $\bu\in\cD_\kappa$, and $\bdu,\bdv^{(\ell)}=(\dv_1^{(\ell)},\dv_2^{(\ell)})\in \ell^1(\Omega)$, $\ell=1,\ldots,m$. For any $1\leq p_1,\ldots,p_m\leq \infty$ such that $\sum_{\ell=1}^m\frac1{p_\ell}=1$, we have that
		\begin{align}
			|\delta^m\deinter_j(\bu)[\bdv^{(1)},\ldots,\bdv^{(m)}]| \, \lesssim\, 
   \max\Big\{1,\frac{|\Gamma_j|}{|\Gamma_{3-j}|}\Big\}
            \Minter_{m,3-j} \prod_{\ell=1}^m\|\bdv^{(\ell)}\|_{\ell^{p_\ell}(\Omega)} \, .
		\end{align}
	\end{lemma}

	\begin{proposition}
		Assume that $\Vmono_{j,\alpha_j}$ and $\Vinter_{j,\alpha_j}$ satisfy the decay conditions \ref{ass-decay}. Assume that $\bu\in \cK_\kappa$ and $\du_j\in c_0(\Omega_j)$\footnote{Recall that $c_0$ is the set of all finite sequences over $\Omega_j$.}, $j=1,2$.
		\begin{enumerate}
			\item \begin{enumerate}[(i)]
				\item $\sitepotmono_{j}[u_j],\sitepotinter_j[\bu]\in L^1(\MoireCell)$.
				\item $\etot\in C^k(\cK_\kappa)$.
			\end{enumerate}
			\item \begin{enumerate}[(i)]
				\item $\dsitepotmono_j[\du_j;u_j], \dsitepotinter_j[\bdu;\bu]\in \ell^1(\cR_j)$. 
				\item There exists a unique continuous extension  $\detot_{\bu}(\cdot) :\ell^1(\Omega_1\times\Omega_2)\to\R$.
				\item $\detot_{\bu}(\cdot)\in C^k(\ell^1(\Omega_1\times\Omega_2))$.
			\end{enumerate}
		\end{enumerate}
	\end{proposition}
	
	\begin{proof}
		The proof essentially follows that of \cite[Section 2.4, Theorem 1]{ortnertheil2013}.
	\end{proof}
	
	Next, we recall from \cite[Theorem 2.17]{hott2023incommensurate} the ergodic theorem resulting in $\etot$.
	
	\begin{proposition}[Ergodic theorem]
		Let $\bu\in\cK_\kappa$. Then we have that
		\begin{align}
			\lim_{N\to\infty}\sum_{j=1}^2\frac1{(2N+1)^2|\Gamma_j|}\sum_{R_j\in\cR_j(N)}\big(\sitepotmono_j[u_j](R_j)+\sitepotinter_j[u_j](R_j)\big) \, = \, \etot(\bu) \, .
		\end{align}
	\end{proposition}
	
	\subsection{Cauchy-Born approximation} In addition to the many-body potentials, we also introduce the Cauchy-Born approximation. For simplicity, we restrict to a single lattice degree of freedom per layer. For an analysis of the continuum/Cauchy-Born approximation in the case of multilattices, we refer the interested reader to \cite{ortner2018CB-multilattice}.
	
	\begin{align}
		\CB_j(u_j) \, := \, \mavint\dx{x} \frac{\Vmono_j
			\big((\nabla u_j(x)\cdot R_j)_{R_j\in\cR_j}\big)}{|\Gamma_j|} \, .	\end{align}
	
	\begin{lemma}
		Let $j\in\{1,2\}$, $m\in\N$. Assume that $\bu=(u_1,u_2)\in\cK_\kappa$ and $v_\ell\in \Wbper^{1,2}(\MoireCell;(\R^2)^{\cA_j})$, $\ell=1,\ldots,m$. For any $1\leq p_1,\ldots,p_m\leq \infty$ such that $\sum_{\ell=1}^m\frac1{p_\ell}=1$, we have that
		\begin{align}
			|\delta^m\CB_j(u_j)[v_1,\ldots,v_m]| \, \lesssim \, \Mmono_{m,j} \prod_{\ell=1}^m\|\nabla v_j\|_{L^{p_\ell}(\MoireCell)} \, .
		\end{align}
	\end{lemma}
	
	\section{Stacking-based relaxation\label{sec-relax}}
	
	Instead of studying relaxation in its full generality, we provide sufficient conditions for the existence of a locally unique minimizer. We will only sketch the proof in the case of a single sublattice degree per layer. We believe that the analysis presented in \cite{ortner2018CB-multilattice} can be adapted to the present context to generalize our results to the case of multilattices. However, for the purpose of clarity and to better emphasize the main ideas, namely the difference between stacking- and phonon-based relaxation, we shall omit this generalization.

	\subsection{Relaxation in Cauchy-Born approximation}
	
	We start by proving the existence of energy minimizers in the Cauchy-Born approximation. The reason for that is that we can employ elliptic regularity theory to prove sufficient regularity for minimizers, in order to employ perturbation theory after that. Here, we treat the interlayer energy contribution as a small perturbation.
	
	The following statement about $\delta^2 \CB_j$ is given by \cite[Proposition 4.1]{Cazeaux-Massatt-Luskin-ARMA2020}. The statement about $\delta^2\mono_j$ can be proved analogously, see, e.g., \cite[Proposition 4]{ortnertheil2013}.
	
	\begin{lemma}\label{lem-mono-coercive}
		For any $v\in \Wper^{1,2}(\MoireCell;(\R^2)^{\cA_j})$, we have that
		\begin{align}
			\delta^2\CB_j(0)[v,v] \, \geq \, \nu_{\mathrm{Stack}}\|\nabla v\|_2 \, .
		\end{align}
		Moreover, there exists $\kappa_1>0$ such that for all $\kappa\in(0,\kappa_1)$, $\bu\in\cK_\kappa$, $j=1,2$, we find that
		\begin{equation}
			\delta^2\CB_j(u_j)[v,v], \, \delta^2\mono_j(u_j)[v,v] \, \geq \, \frac{\nu_{\mathrm{Stack}}}2\|\nabla
			v\|_2 \, .
		\end{equation}
	\end{lemma}

	\begin{proposition}[Stacking-based stability of displaced lattices]\label{prop-relaxed-stability}
		Let $\kappa_1$ be as in Lemma \ref{lem-mono-coercive}. Then there exists a constant $\delta_0>0$ dependent on $|\Gamma_1|$, $|\Gamma_2|$ such that 
		\begin{align}
			\Minter_2 \, \leq \, 
            \frac{\nu_{\mathrm{Stack}}}{\max\Big\{\frac{|\Gamma_1|}{|\Gamma_2|},\frac{|\Gamma_2|}{|\Gamma_1|}\Big\}|\MoireCell|}\delta_0 
		\end{align}
		implies that, for all $\kappa\in(0,\kappa_1)$, $\bu\in\cK_\kappa$ and $\bv\in \Wbper^{1,2}$, we have that
		\begin{align}
			\delta^2\etot(\bu)[\bv,\bv] \, \geq \, \frac{\nu_{\mathrm{Stack}}}2\|\nabla \bv\|_2^2 \, .
		\end{align} 
	\end{proposition}
	
	\begin{proof}
		By Lemma \ref{lem-inter-Cm-bd}, we have that
		\begin{align}
			\delta^2\inter_j(\bu)[\bv,\bv] \, &\lesssim\, \max\Big\{1,\frac{|\Gamma_j|}{|\Gamma_{3-j}|}\Big\}
            \Minter_{2,3-j} \|\bv\|_2^2\\
			&\lesssim \max\Big\{1,\frac{|\Gamma_j|}{|\Gamma_{3-j}|}\Big\}|\MoireCell| \, 
   \Minter_{2,3-j} \|\nabla\bv\|_2^2 \, ,
		\end{align}
		where, in the second step, we employed Poincar\'e's inequality and the fact that $\mavint\dx{x}\bv(x)=0$. In particular, there exists a constant $\delta_0>0$ dependent on $|\Gamma_1$, $|\Gamma_2|$ such that if
		\begin{align}
			\Minter_2 \, \leq \, \frac{\nu_{\mathrm{Stack}}}{\max\Big\{\frac{|\Gamma_1|}{|\Gamma_2|},\frac{|\Gamma_2|}{|\Gamma_1|}\Big\}|\MoireCell|}\delta_0 
		\end{align}
		then 
		\begin{align}
			\delta^2\etot(\bu)[\bv,\bv] \, \geq \, \frac{\nu_{\mathrm{Stack}}}2\|\nabla \bv\|_2^2 \, .
		\end{align}
		This concludes the proof.
	\end{proof}
	
	Next, we want to prove the existence of a local minimizer. For that, we employ the following generalized inverse function theorem as stated, e.g., in \cite[Lemma 14]{ortnertheil2013}.
	
	\begin{lemma}\label{lem-INV-FT}
		Let $X,Y$ be Banach spaces, $O$ an open subset of $X$, and let $T:O\to Y$ be Fr\'echet differentiable. Assume that $x_0\in O$ satisfies 
		\begin{equation}
			\|T(x_0)\|_Y \, \leq\, \eta, \quad \|\delta T(x_0)^{-1}\|_{\cL(Y,X)} \, \leq \, \sigma , \quad \overline{B_X(x_0,2\eta\sigma)}\subseteq O \, .
		\end{equation}
		Assume that $\delta T$ satisfies the Lipschitz bound $$\|\delta T(x_1)-\delta T(x_2)\|_{\cL(X,Y)}\leq L\|x_1-x_2\|_X \quad \forall x_1,x_2\in \overline{B_X(x_0,2\eta\sigma)},$$
		and that $2L\sigma^2\eta<1$. Then there exists a unique $x\in B_X(x_0,2\eta\sigma)$ such that $T(x)=0$.
	\end{lemma}
	
	Next, we first prove the existence of local minimizers for when the monolayer energy is approximated by the Cauchy-Born energy. This in turn will allow us to prove the existence of a local minimizers for the non-approximated monolayer energy in proximity of the Cauchy-Born minimizer. This strategy was also pursued in \cite{ortnertheil2013}. In particular, there are cases when the solution to the problem in the Cauchy-Born approximation is stable while the atomistic model is unstable, see \cite{wei2007cauchy,hudson2012stability,ortnertheil2013}. 
	\par For this purpose, we introduce the Cauchy-Born approximation of the total energy
	\begin{align}\label{def-etotCB}
		\etotCB(\bu) \, := \, \sum_{j=1}^2\big(\CB_j(u_j)+\inter_j(\bu)\big) \, .
	\end{align}

	\begin{theorem}\label{thm-relax-CB}
		Let $\delta_0>0$ be as in Proposition \ref{prop-relaxed-stability}. Then there exists a constant $\delCB\in(0,\delta_0)$ dependent on $|\Gamma_1$, $|\Gamma_2|$ such that if
		\begin{align}
			\Minter_2 \, \leq \, \frac{\nu_{\mathrm{Stack}}}{\max\Big\{\frac{|\Gamma_1|}{|\Gamma_2|},\frac{|\Gamma_2|}{|\Gamma_1|}\Big\}|\MoireCell|} \delCB ,
		\end{align}
		and that
		\begin{align}
			\MoveEqLeft\max_{j=1,2}\|\disregj\|_2\Minter_2\Big(1+\nu_{\mathrm{Stack}}^{-1}\max_{j=1,2}\|\disregj\|_2^2\sum_{\ell=1}^3\Minter_\ell\Big)^2\\
			&\Big(\sum_{\ell=3}^5\Mmono_\ell+\sum_{\ell=3}^4\Minter_\ell\Big)\leq \delCB \, ,
		\end{align}
		then there exists a solution $\buCB\in\Wbper^{3,2}$ of
		\begin{align}
			\delta\etotCB(\buCB)[\bv] \, = \, 0 \quad \forall \bv \in \Wbper^{3,2} \, .
		\end{align}
		This solution is \emph{stable} in the sense that
		\begin{align}
			\delta^2\etotCB(\buCB)[\bv,\bv] \, \gtrsim \, \|\nabla\bv\|_2^2 \quad \forall \bv \in\Wbper^{1,2} \, .
		\end{align}
		It also satisfies the a priori bound
		\begin{align}
			\MoveEqLeft\|\buCB\|_{W^{3,2}} \, \lesssim \, \max_{j=1,2}\|\disregj\|_2\Minter_2\Big(1+\nu_{\mathrm{Stack}}^{-1}\max_{j=1,2}\|\disregj\|_2^2\sum_{\ell=1}^3\Minter_\ell\Big) \, .
		\end{align}
	\end{theorem}

	\begin{proof}
		The proof is standard and is only a slight modification of the arguments presented in the proof of \cite[Theorem 4.3]{Cazeaux-Massatt-Luskin-ARMA2020}.
		\par We want to apply Lemma \ref{lem-INV-FT} to $X=\Wbper^{3,2}$, $Y=\Wbper^{1,2}$, $x_0=0$, $O=B_X(0,\kappa')$, where $\kappa'>0$ is chosen small enough such that, due to the Sobolev embedding $\Wbper^{3,2}\subseteq \Wbper^{1,\infty}$, $B_X(0,\kappa')\subseteq K_{\kappa}$ for $\kappa\in(0,\kappa_1$ and $\kappa_1>0$ chosen as in Lemma \ref{lem-mono-coercive}. We define the mapping
		\begin{align}
			T:O\to Y , \, \bu\mapsto \delta\etotCB(\bu) \, .
		\end{align}
		\begin{enumerate}
			\item {\it Residual bound:} Observe that, due to
			\begin{equation}
				\delta\CB_j(0)[v] \, \propto \, \mavint\dx{x} v(x) \, = \, 0
			\end{equation}
			for $v\in \Wbper^{1,2}$, we have that 
			\begin{align}
				T({\bf 0})(x) \, =& \, \sum_{j=1}^2\delta\inter_j({\bf 0})(x)\\
				=& \,\sum_{\ell,j=1}^2\Big[\sum_{R_{3-j}}(\Vinter_{j;R_{3-j}})_\ell\big(\Delta_{j,3-j}[{\bf 0}](x)\big)\\
				& \, -\sum_{R_j}(\Vinter_{3-j;R_j})_\ell\big(\Delta_{3-j,j}[{\bf 0}]\big(A_{3-j}A_j^{-1}(x-R_j)\big)\big)\Big]e_{\ell+j}, \label{eq-T0-expression}
			\end{align} 
			where $e_k$ is the standard unit vector in the $k$-the direction, and where we recall \eqref{eq-ELE-MB-cont}. Consequently, we obtain that
			\begin{equation}
				\|T({\bf 0})\|_Y \, \leq \, \sum_{j=1}^2\|\delta\inter_j({\bf 0})\|_{W^{1,2}}  \, \lesssim \, \sum_{j=1}^2\|\nabla\delta\inter_j({\bf 0})\|_2 \, , 
			\end{equation}
			where, in the last step, we used Poincar\'e's inequality. Plugging in the expression \eqref{eq-T0-expression}, we thus obtain
			\begin{equation}\label{eq-T0-bd}
				\|T({\bf 0})\|_Y \, \lesssim \, \max_{j=1,2}\|\disregj\|_2\Minter_2 \, .
			\end{equation}
			By assumption, we thus can choose $\|T({\bf 0})\|_Y$ small enough.
			
			\item \emph{Stability estimate:} Proposition \ref{prop-relaxed-stability} implies that $\delta T(0)^{-1}\in \cL(Y,Y)$, e.g., by the Lax-Milgrams theorem. We adapt the proof of \cite[Theorem 4.3]{Cazeaux-Massatt-Luskin-ARMA2020} to the present case to prove that even $\delta T(0)^{-1}\in \cL(Y,X)$. 
			\par \hspace*{2ex} For that, let $\bv\in Y$. Due to the stability estimate described above, there then exists a unique solution $\bu\in Y$ to $\delta T(0)[\bu]=\bv$ with $\|\bu\|_Y\leq \frac2{\nu_{\mathrm{Stack}}}\|\bv\|_Y$. Let $\bw:=-\sum_{j=1}^2\delta^2\inter_j({\bf 0})[\bu,\cdot]$. A straight-forward but lengthy computation, analogously to that leading to \eqref{eq-T0-bd} yields 
			\begin{align}
				\|\bw\|_{W^{1,2}} \, \lesssim \, \max_{j=1,2}\|\disregj\|_2^2\sum_{\ell=1}^3\Minter_\ell\|\bu\|_{W^{1,2}} \, . 
			\end{align}
			Then $\bu$ solves the elliptic problem with constant coefficients
			\begin{equation}
				\sum_{j=1}^2\delta^2\CB_j({\bf 0})[\bu,\cdot] \, = \, \bv+\bw \, .
			\end{equation}
			Elliptic regularity theory as explained, e.g., in \cite{Evans-PDE-book,gilbarg-trudinger-book}, then yields that $\bu\in \in X$, with the a priori estimate $\|\bu\|_{W^{3,2}}\lesssim \|\bv+\bw\|_{W^{1,2}}$. In particular, $\delta T(0):X\to Y$ is an isomorphism with the stability estimate
			\begin{align}
				\|\delta T(0)^{-1}\|_{\cL(Y,X)} \, \lesssim \, \Big(1+\nu_{\mathrm{Stack}}^{-1}\max_{j=1,2}\|\disregj\|_2^2\sum_{\ell=1}^3\Minter_\ell\Big) \, .
			\end{align}
			
			\item \emph{Lipschitz bound:} A lengthy but straight-forward calculation, employing the Sobolev embedding $\|\bu\|_{1,\infty}\lesssim\|\bu\|_{2,4}\lesssim\|\bu\|_{3,2}$ implies that 
			\begin{align}\label{eq-CB-Lip}
				\MoveEqLeft\|\delta^2\CB_j(\bu)[\bv,\cdot]-\delta^2\CB_j(\bu')[\bv,\cdot]\|_{W^{1,2}}\\ 
				& \lesssim \, \sum_{\ell=3}^5\Mmono_{\ell,j}\|\bu-\bu'\|_{W^{3,2}}\|\bv\|_{W^{3,2}} \,.
			\end{align}
			Analogously, one proves that
			\begin{align}\label{eq-inter-Lip}
				\MoveEqLeft\|\delta^2\inter_j(\bu)[\bv,\cdot]-\delta^2\inter_j(\bu')[\bv,\cdot]\|_{W^{1,2}}\\ 
				& \lesssim \, \sum_{\ell=3}^4\Minter_{\ell,j}\|\bu-\bu'\|_{W^{3,2}}\|\bv\|_{W^{3,2}} \,.
			\end{align}
			With a computation analogous to \cite[Proof of Theorem 3]{ortnertheil2013} and similar to those proving \eqref{eq-CB-Lip}, \eqref{eq-inter-Lip}, one can also prove that
			\begin{align}
				\MoveEqLeft\|T(\bu+\bv)-T(\bu)-\delta^2\etot(\bu)[\bv,\cdot]\|_{W^{1,2}} \\
				&\lesssim \, \Big(\sum_{\ell=3}^5\Mmono_\ell+\sum_{\ell=3}^4\Minter_\ell\Big)\|\bv\|_{W^{3,2}}^2 \, .
			\end{align}
			In particular, $T:X\to Y$ is Fr\'echet differentiable and its derivative $\delta^2\etot$ satisfies the Lipschitz bound
			\begin{align}
				\|\delta T(\bu)-\delta T(\bu')\|_{\cL(X,Y)} \, \lesssim \, \Big(\sum_{\ell=3}^5\Mmono_\ell+\sum_{\ell=3}^4\Minter_\ell\Big)\|\bu-\bu'\|_{W^{3,2}} \, .
			\end{align}
		\end{enumerate}
		Consequently, Lemma \ref{lem-INV-FT} implies the existence of a local minimizer $\buCB$ and its a priori bound. The stability bound follows from Proposition \ref{prop-relaxed-stability}.
	\end{proof}
	
	\subsection{Relaxation in atomistic model}
	
	Our goal is now to show that existence of a minimizer for the full atomistic model which is a neighborhood of a minimizer for the Cauchy-Born approximation. This will require to impose additional decay assumptions on $\Vmono$. In particular, these assumptions are sufficient to justify the existence of a minimizer for the atomistic model.

	The next statement employs \cite[Lemma 12]{ortnertheil2013}. However, stating it properly would require us to introduce some notation that is similar to that in the present work. In order to avoid confusion and to keep the presentation brief, we thus omit an extensive introduction, and refer the interested reader to the original reference instead.

	\begin{theorem}\label{thm-relax-atomistic}
		Let $\delCB>0$ be given as in Theorem \ref{thm-relax-CB}. Assume that $\bUCB$ is a local minimizer of $\etotCB$ and, for any $\atd>0$, define 
		\begin{equation}
			\buCB(x) \, := \, \atd^{-1}\bUCB(\atd x) \quad \forall x\in \MoireCellres \, := \, \frac1\atd\MoireCell \, .
		\end{equation} 
        Then there exist 
        there exists a constant $\delmono\in(0,\delCB)$ dependent on $|\Gamma_1$, $|\Gamma_2|$ such that if 
        \begin{align}
            \MoveEqLeft\atd \Minter_2 \frac{\max_{j=1,2}\|\disregj\|_2}{\nu_{\mathrm{Stack}}^2}\Big(1+\frac{\max_{j=1,2}\|\disregj\|_2^2}{\nu_{\mathrm{Stack}}}\sum_{\ell=1}^3\Minter_\ell\Big)\\
            & \sum_{\ell=2}^3\Mmonoat_{\ell,2}\Big(\Mmono_3+ \max\Big\{\frac{|\Gamma_1|}{|\Gamma_2|},\frac{|\Gamma_2|}{|\Gamma_1|}\Big\}
            \Minter_3\Big) \, \leq \, \delmono 
        \end{align}
        and
        \begin{align}
			\Minter_2 \, \leq \, \frac{\nu_{\mathrm{Stack}}}{\max\Big\{\frac{|\Gamma_1|}{|\Gamma_2|},\frac{|\Gamma_2|}{|\Gamma_1|}\Big\}|\MoireCell|}
   \delCB 
		\end{align}
		and 
		\begin{align}
			\MoveEqLeft\max_{j=1,2}\|\disregj\|_2\Minter_2\Big(1+\frac{\max_{j=1,2}\|\disregj\|_2^2}{\nu_{\mathrm{Stack}}}\sum_{\ell=1}^3\Minter_\ell\Big)^2\\
			&\Big(\sum_{\ell=3}^5\Mmono_\ell+\sum_{\ell=3}^4\Minter_\ell\Big)\leq \delCB \, ,
		\end{align}
        there exists a stable local minimizer $\bueq$ of $\etot$ which satisfies
		\begin{align}
			\atd\|\nabla \bueq-\nabla\buCB\|_{L^2(\MoireCellres)} \, \lesssim \, &\atd \Minter_2 \frac{\max_{j=1,2}\|\disregj\|_2}{\nu_{\mathrm{Stack}}}\sum_{\ell=2}^3\Mmonoat_{\ell,2}\\
            &\,\Big(1+\frac{\max_{j=1,2}\|\disregj\|_2^2}{\nu_{\mathrm{Stack}}}\sum_{\ell=1}^3\Minter_\ell\Big) \, .
		\end{align}
	\end{theorem}
	
	\begin{proof}
        We want to apply Lemma \ref{lem-INV-FT} for the case $x_0 =\buCB$, $X=\Wbper^{1,2}$, endowed with the $\dot{W}^{1,2}$ norm, $Y=\Wbper^{-1,2}$, endowed with the $\dot{W}^{-1,2}$ norm, $O:=B_{\Wbper^{1,2}}(0,\eta_0)$. Let $\kappa_1>0$ be as in Lemma \ref{lem-mono-coercive}. We choose $\eta_0$ small enough s.t. $B_{\Wbper^{1,2}}(x_0,\eta_0)\subseteq \Wbper^{1,2}\cap\cD_{\kappa_1}$. Then we define
		\begin{align}
			T(\bw) \, := \, \delta \etot(\buCB+\bw) \, : \, \Wbper\to\R \, .
		\end{align}
		Lemmata \ref{lem-mono-db} and \ref{lem-inter-Cm-bd} imply that $T$ is Fr\'echet differentiable on $O$ and that $\delta T$ is Lipschitz continuous on $O$ with Lipschitz constant 
		\begin{align}
			C \Big(\Mmono_3+\max\Big\{\frac{|\Gamma_1|}{|\Gamma_2|},\frac{|\Gamma_2|}{|\Gamma_1|}\Big\}\Minter_3\Big)    
		\end{align}
		for some universal constant $C>0$.
		\par The stability estimate follows immediately from Proposition \ref{prop-relaxed-stability}. In particular, we have that
		\begin{align}
			\|\delta T(0)^{-1}\|_{\cL(Y,X)} \, \leq \, \frac2{\nu_{\mathrm{Stack}}} \, .
		\end{align}
		\par Next, we want to establish a residual bound. Given $\bv\in \Wbper$ and recalling \eqref{def-quasi-interpolant}, we abbreviate
        \begin{align}
            \widetilde{\bv} \, := \, (\tilde{v}_j)_{j=1}^2 \, , \quad \tilde{v}_j(x) \, := \, \zeta*v_j(x) \, = \, \int_{\MoireCellres}\dx{y}\zeta(x-y)v_j(y) \, ,
        \end{align}
        where we extended $v_j$ trivially by $0$ outside of $\MoireCell$.
        \par Observe that 
		\begin{align}\label{eq-T0-mono-relax-exp-1}
			T(0)[\widetilde{\bv}] \, = \, \delta\etot(\buCB)[\widetilde{\bv}] \, = \, \sum_{j=1}^2\big(\delta\mono_j(\uCB_j)[\tilde{v}_j]-\delta\CB_j(\uCB_j)[v_j]\big) \, ,
		\end{align}
		where we used the fact that $\etotCB(\buCB)=0$. Recalling the localization formula \eqref{eq-localization}, we have that
        \begin{align}
            \delj_{R_j} \tilde{v}_j \, = \, \int_{\MoireCellres}\dx{y}\underbrace{\frac1{|\Gamma_j|}\int_0^1\dx{t}\zeta_j(x+tR_j-y)}_{=:\zeta_{j,x,R_j}(y)}\nabla_{R_j}v(y)
        \end{align}
        we define, as in \cite{ortnertheil2013}, the atomistic stress density tensor 
        \begin{equation}\label{def-smono}
            \smono_j(u;y) \, := \, \frac1{|\Gamma_j|}\mavresint\dx{x}\sum_{R_j}\big[\Vmono_{j;R_j}\big(\Delta_j[u](x)\big)\otimes R_j\big]\zeta_{j,x,R_j}(y) \, ,
        \end{equation}
        which allows us to compute that 
		\begin{align}
			\delta\mono_j(\uCB_j)[\zeta_j*v_j] \, &= \, \frac1{|\Gamma_j|}\mavresint\dx{x}\sum_{R_j}\Vmono_{j,R_j}\big(\Delta_j[\uCB_j](x)\big)\cdot \delj_{R_j}(\zeta_j*v_j)(x)\\
			&= \, \int_{\MoireCellres}\dx{y}\smono_j(\uCB_j;y):\nabla v_j(y) \, .
		\end{align}
        Similarly, abbreviating the first Piola–Kirchhoff stress density tensor
        \begin{align}
            \sPK_j(u;y) \, := \, \frac1{|\Gamma_j|}\sum_{R_j\in\cR_j}\Vmono_{j;R_j}\big((\nabla u(x)\cdot R_j)_{R_j\in\cR_j}\big)\otimes R_j \, ,
        \end{align}
        we can rewrite
        \begin{align}
            \delta\CB_j(\uCB_j)[v_j] \, = \,\int_{\MoireCell}\dx{y}\sPK_j(\uCB_j;y):\nabla v_j(y) \, .
        \end{align}
        Consequently, \eqref{eq-T0-mono-relax-exp-1} implies the estimate
        \begin{align}
            |T(0)[\widetilde{\bv}]| \, \leq \, \sum_{j=1}^2\|\smono_j(\uCB_j)-\sPK_j(\uCB_j)\|_{L^2(\MoireCellres)}\|\nabla v_j\|_{L^2(\MoireCellres)} \, .
        \end{align}
        An analogous computation to \cite[Proposition 3]{ortnertheil2013} now implies that, given $\bv\in\Wbper$ with $\|\nabla \bv\|_{L^2(\MoireCellres)}=1$, 
        \begin{align}
            |T(0)[\widetilde{\bv}]| \, &\lesssim \, \big(\Mmonoat_{2,2}\|\nabla^3\buCB\|_{L^2(\MoireCellres)}+\Mmonoat_{3,2}\|\nabla^2\buCB\|_{L^4(\MoireCellres)}^2\big)\\
            &\lesssim \, \atd\big(\Mmonoat_{2,2}\|\nabla^3\bUCB\|_{L^2(\MoireCell)}+\Mmonoat_{3,2}\|\nabla^2\bUCB\|_{L^4(\MoireCell)}^2\big) \, ,
        \end{align}
        where we used the definition of $\buCB$ and the scaling properties of the involved norms. Recalling the a priori bound in Theorem \ref{thm-relax-CB} and employing the fact that $\|\nabla \tilde{v}_j\|_2\sim \|\nabla v_j\|_2$, we thus obtain the upper bound
        \begin{align}
            \|T(0)\|_Y \, \lesssim \, & \atd\sum_{\ell=2}^3\Mmonoat_{\ell,2}\max_{j=1,2}\|\disregj\|_2\Minter_2\\
            & \Big(1+\frac{\max_{j=1,2}\|\disregj\|_2^2}{\nu_{\mathrm{Stack}}}\sum_{\ell=1}^3\Minter_\ell\Big) \, .
        \end{align}
        In particular, there exists a universal constant $\delmono>0$ such that if we thus choose $\atd>0$ such that
        \begin{align}
            \MoveEqLeft\atd \Minter_2 \frac{\max_{j=1,2}\|\disregj\|_2}{\nu_{\mathrm{Stack}}^2}\Big(1+\frac{\max_{j=1,2}\|\disregj\|_2^2}{\nu_{\mathrm{Stack}}}\sum_{\ell=1}^3\Minter_\ell\Big)\\
            & \sum_{\ell=2}^3\Mmonoat_{\ell,2}\Big(\Mmono_3+\max\Big\{\frac{|\Gamma_1|}{|\Gamma_2|},\frac{|\Gamma_2|}{|\Gamma_1|}\Big\}
            \Minter_3\Big) \, \leq \, \delmono \, ,
        \end{align}
        Lemma \ref{lem-INV-FT} implies that there exists a locally unique solution $\bv$ of $T[\bv]=0$. We define $\bueq:=\buCB+\bv$. Lemma \ref{lem-INV-FT} yields
        \begin{align}
            \|\nabla\bueq-\nabla\buCB\|_2 \, = \, \|\bv\|_2 \, \lesssim \, &\atd \Minter_2 \frac{\max_{j=1,2}\|\disregj\|_2}{\nu_{\mathrm{Stack}}}\sum_{\ell=2}^3\Mmonoat_{\ell,2}\\
            &\,\Big(1+\frac{\max_{j=1,2}\|\disregj\|_2^2}{\nu_{\mathrm{Stack}}}\sum_{\ell=1}^3\Minter_\ell\Big) \, .
        \end{align}
        This finishes the proof.
	\end{proof}
	
	\section{Phonon stability\label{sec-phonon-stab}}

    As we have previously seen, a critical point for the stacking-based energy density $\etot$ yields a critical point for the offset energy $\detot$. A natural question that we then need to answer is whether stability is also inherited. Recall that an important property of $\etot$ is that it is stable in a neighborhood of zero displacement $\bu={\bf 0}$, which allowed us to obtain a locally unique minimizer.
    
    \par As we shall see, stability is not necessarily inherited in the following sense: First, we will see that $\delta^2\detot_{\bf 0}(0)<0$ despite $\delta^2\etot (0)>0$. After developing a stability criterion, we then show that for purely twisted bilayer systems, at twist angle tending to $0$, the minimizer $\bueq_{\theta\to0}$ does indeed satisfy $\delta^2\detot_{\bueq_{\theta\to0}}(0)>0$.

    \par For the next assumption, recall the interpolants \eqref{def-fo-interpol}--\eqref{def-quasi-interpolant} of functions defined on the lattices $\cR_j$. 
 
	\begin{assumption}[Phonon stability of monolayer]
		We assume that 
		\begin{equation*}
			\nu_{\mathrm{Phonon}} \, := \, \inf_{\|\nabla v\|_2=1}\delta^2\demono_j[0][v,v] \, > \, 0 \, .
		\end{equation*}
	\end{assumption}
	In particular, we have that
	\begin{align*}
		\delta^2 \demono_{j;\du_j}(\du_j) \, \geq \, \nu \|\nabla \du_j\|_2^2 \, .
	\end{align*}
	For the interlayer potential, we consider the simplified problem for interatomic pair potentials $v$, and with only one sublattice degree of freedom. In this case, we have
	\begin{align*}
		\delta^2 \deinter_{j;\bueq}[{\bm 0}][\bdu,\bdv] \, &= \, \frac12 \sum_{R_1,R_2}(D^2v)(R_1+\ueq_1(R_1)-R_2-\ueq_2(R_2))\\
		&\qquad [\du_1(R_1)-\du_2(R_2),\dv_1(R_1)-\dv_2(R_2)] \, .
	\end{align*}

	\subsection{Continuum approximation\label{sec-cont-approx}}
	
	We now want to study the continuum approximation, in order to find an appropriate test function to approximate the equilibrium energy. 
 
    \par We start by \emph{excluding} the relaxation and fixing the the equilibrium distance $d$. We focus on the interlayer energy in the case of pair potentials
	\begin{equation}
		\sum_{R_1,R_2} (v_1(R_1)-v_2(R_2))^T D^2v(R_1-R_2,d) (v_1(R_1)-v_2(R_2)) \, .
	\end{equation}
	Observe that we can rewrite
	\begin{align*}
		\sum_{R_1,R_2}h(R_1,R_2) \, = \, \int \dxx{x}\dx{y} \sum_{R_1}\delta(x-R_1)\sum_{R_2}\delta(y-R_2) h(x,y) \, .
	\end{align*}
	Notice that we can approximate $\delta(x-R_j)$ by $\zeta_j(x-R_j)$. Abbreviating the approximate lattice restrictions
	\begin{equation}\label{def-chi-Rj}
		\chi_{\cR_j}(x) \, := \, \sum_{R_j\in\cR_j}\zeta_j(x-R_j) \, ,
	\end{equation} 
	and the approximate interaction potential
	\begin{equation*}
		A(x,y) \, := \, \chi_{\cR_1}(x)\chi_{\cR_2}(y)(D^2v)(x-y,d)
	\end{equation*}
	we thus approximate $\delta^2 \deinter_{j;\bueq}[{\bm 0}][\bdu,\bdu]$ by
	\begin{align*}
		\int\dxx{x}\dx{y} (\du_1(x)-\du_2(y))^TA(x,y)(\du_1(x)-\du_2(y)) \, .
	\end{align*}
	For simplicity, we even consider
	\begin{equation*}
		A(x,y) \, \leftarrow \,  D^2v(x-y,d) \, .
	\end{equation*}
	Formally, we can justify this simplification by rescaling the interlayer energy by the unit cell volume. In particular, we rescale and take the continuum limit
	\begin{equation}
		\chi_{\cR_j}(x) \, \rightarrow \, \atd^2 \sum_{R_j\in\atd\cR_j}\zeta_j(x-R_j) \xrightarrow{\atd\to0} \int \dx{y}\zeta_j(x-y) \, = \, 1 \, ,
	\end{equation}
	by the normalization condition on $\zeta_j$.
	
	We expand the interlayer energy according to
	\begin{align}
		\label{eq-inter-decomp}
		&\int \dxx{x}\dx{y} (v_1(x)-v_2(y))^T D^2v(x-y,d) (v_1(x)-v_2(y)) \, =\\
		& \quad \int \dxx{x}\dx{y} (v_1(x)-v_1(y))^T D^2v(x-y,d) (v_1(x)-v_1(y)) \label{eq-inter-decomp-mono}\\
		& \quad + \int \dxx{x}\dx{y} (v_1(y)-v_2(y))^T D^2v(x-y,d) (v_1(y)-v_2(y))\label{eq-inter-decomp-inter-av}\\
		& \quad +2 \int \dxx{x}\dx{y} (v_1(x)-v_1(y))^T D^2v(x-y,d) (v_1(y)-v_2(y)) \, . \label{eq-inter-decomp-CSI}
	\end{align}
	This allows us to make the following observations:
	\begin{enumerate}
		\item \eqref{eq-inter-decomp-mono} corresponds to a \emph{monolayer} contribution in layer 1. In particular, if $v$ satisfies some relative compactness property w.r.t. to the monolayer energy, we can treat this contribution perturbatively.
		\item For \eqref{eq-inter-decomp-inter-av}, we notice that $v_1$ and $v_2$ are evaluated at the \emph{same} position, which allows us to \emph{average} over $D^2v$. Negative eigenvalues of this average could lead to negative total energies, dependent on the eigenfunctions.
		\item \eqref{eq-inter-decomp-CSI} can be controlled via Cauchy-Schwarz' inequality by distributing it as corrections to the two previous contributions coming from \eqref{eq-inter-decomp-mono} and \eqref{eq-inter-decomp-inter-av}.
	\end{enumerate}

	As an interlayer potential, we fix a potential that is radial in the horizontal direction, i.e.,
	\begin{equation}
		v(x,z) \, = \, w(|x|,z) \, ,
	\end{equation}
	where $x\in\R^2$. A straight-forward calculation yields
	\begin{align}
		D^2v(x,z) \, = \, \begin{pmatrix}
			\frac{\partial_\rho w(\rho,z)}\rho Q_x \, + \, \partial_\rho^2w(\rho,z) P_x & \partial_{\rho z}^2w(\rho,z) \frac{x}\rho \\
			\partial_{\rho z}^2w(\rho,z) \frac{x^T}\rho & \partial_z^2 w(\rho,z)
		\end{pmatrix}  \, ,
	\end{align}
	where $P_x (y^T,z)^T=\frac{\scp{x}{y}}{|x|}$ and $Q_x=I-P_x$. 
	Another straight-forward computation implies that 
	\begin{align}
		\int_{\R^2} \dx{x}D^2v(x,z) \, &= \, \pi \int_0^\infty \dxx{\rho} \begin{pmatrix}
			\big(\rho\partial_\rho^2w(\rho,z)+\partial_\rho w(\rho,z)\big)I_2 & 0\\
			0 & 2\rho\partial_z^2w(\rho,z)
		\end{pmatrix}\\
		&= \, \pi \begin{pmatrix}
			\rho\partial_\rho w(\rho,z)\Big|_{\rho=0}^{\rho=\infty} & 0\\
			0 & 2\int_0^\infty \dx{\rho}\rho \partial_z^2w(\rho,z)
		\end{pmatrix}\\
		&=: \, \diag(0,0,m(z)) \, , \label{eq-D2v-av}
	\end{align}
	where, in the second line, we integrated by parts and used the fact that $w(\rho,\cdot)\lesssim\frac1{\rho^{2+}}$ for $\rho\to\infty$ and $\|\partial\rho w(\cdot,z)\|_\infty<\infty$ for $|z|\neq 0$. 
	We now consider two special cases.
	\begin{enumerate}
		\item $w(\rho,z)=f_1(\rho)f_2(z)$: Then $m(z)=2\pi \Big(\int_0^\infty \dx{\rho}\rho f_1(\rho)\Big)f_2''(z)$.
		\item $w(\rho,z)=g(r)$, where $r=\sqrt{\rho^2+z^2}$: A straight-forward calculation yields
		\begin{align}
			\partial_z^2w(\rho,z) \, &= \, \frac{\rho^2}{r^3}g'(r) + \Big(\frac{z}{r}\Big)^2g''(r) \, = \, \frac{g'(r)}r + g''(r)- \partial_\rho^2 w(\rho,z) \, ,\\
			\frac{\rho}rg'(r) \, & = \, \partial_\rho w(\rho,z) \, , \quad 
			\rho g''(r) \, = \, \frac{z^2}\rho \partial_\rho^2w(\rho,z)-\frac{z^2}{\rho^2}\partial_\rho w(\rho,z) \, .
		\end{align}
		Consequently, we obtain that
		\begin{align}
			m(z) \, &= \, 2\pi \int_0^\infty \Big(\partial_\rho w(\rho,z)+\frac{z^2}\rho \partial_\rho^2w(\rho,z)-\frac{z^2}{\rho^2}\partial_\rho w(\rho,z)- \rho\partial_\rho^2 w(\rho,z)\Big)\\
			&= \, 2\pi\Big[\Big(\frac{z^2}\rho \partial_\rho w(\rho,z)-\rho\partial_\rho w(\rho,z)\Big)\Big|_{\rho=0}^{\rho=\infty}+2\int_0^\infty \dx{\rho} \partial_\rho w(\rho,z)\Big] \, .
		\end{align}
		Employing $\partial_\rho w(\rho,z)=\frac{\rho}rg'(r)$ and the fundamental theorem of calculus, we thus obtain
		\begin{equation}
			m(z) \, = \, 2\pi\big(|z|g'(|z|)+2g(|z|)\big) \, .
		\end{equation}
	\end{enumerate}
	
	\begin{example}[3D Lennard-Jones]\label{ex-LJ}
		We have that
		\begin{equation}
			g_{\mathrm{LJ}}(r) \, = \, 4\vep_0\Big[\Big(\frac{\sigma}{r}\Big)^{12}-\Big(\frac{\sigma}{r}\Big)^6\Big]
		\end{equation}
		for $\vep_0,\sigma>0$. $2^{\frac16}\sigma\approx 1.12\sigma$ is the equilibrium distance. Then we obtain
		\begin{equation}
			m_{\mathrm{LJ}}(z) \, = \, 8\pi \vep_0\Big[5\Big(\frac{\sigma}{|z|}\Big)^6 - 11\Big(\frac{\sigma}{|z|}\Big)^{12}\Big] \, .
		\end{equation}
		$m_{\mathrm{Lj}}$ changes its sign from negative to positive at $|z|=\big(\frac{11}{5}\big)^{\frac16}\sigma\approx 1.14\sigma$.
        \par In the case of bilayer graphene, we have that the equilibrium distance is given by $z=d_{\mathrm{BG}}=3.35 \mbox{\ \emph{\AA}}$. The parameters are given by $\vep_0=2.39\mbox{\ \emph{meV}}$, $\sigma=3.41\mbox{\ \emph{\AA}}$, see \cite{KolmogorovCrespi}. We compute
        \begin{equation}
            m_{\mathrm{LJ}}^{\mathrm{(BG)}} \, = \, -483.5 \mbox{\ \emph{meV}} \, < \, 0 \, .
        \end{equation}
	\end{example}
	
	\begin{example}[3D Morse potential]\label{ex-3dmorse}
		We have that
		\begin{equation}
			g_{\mathrm{Morse}}(r) \, = \, E_0\big(e^{-2\kappa(r-r_0)}-2e^{-\kappa(r-r_0)}\big)
		\end{equation}\begin{equation}
		g_{\mathrm{Morse}}(r) \, = \, E_0\big(e^{-2\kappa(r-r_0)}-2e^{-\kappa(r-r_0)}\big)
		\end{equation}
		for some dissociation energy $E_0>0$ and decay rate $\kappa>0$. Then we obtain
		\begin{equation}
			m_{\mathrm{Morse}}(z) \, = \, -4\pi E_0\Big((\kappa|z|-1)e^{-2\kappa(|z|-r_0)}-(\kappa|z|-2)e^{-\kappa(|z|-r_0)}\Big) \, .
		\end{equation}
		$m_{\mathrm{Morse}}$ changes its sign from negative to positive once in the interval $[0,\min\{\frac2\kappa,r_0\}]$ and from positive to negative once in $[\max\{\frac2\kappa,r_0\},\infty)$, see fig. \ref{fig-Morse-transition}. The transition points are given by the solutions of
		\begin{equation}
			e^{\kappa(|z|-r_0)} -1 \, = \, \frac1{\kappa|z|-2} \, .
		\end{equation}
        In case of bilayer graphene, we have that , see \cite{oconnor2015airebo}, $\kappa= 1.8168\mbox{\ \emph{\AA}}^{-1}$, $E_0=2.8437\mbox{\ \emph{meV}}$, $r_0=3.6891\mbox{\ \emph{\AA}}$ and again $z=d_{\mathrm{BG}}=3.35 \mbox{\ \emph{\AA}}$ as above. Consequently, we obtain
        \begin{equation}
            m_{\mathrm{Morse}}^{\mathrm{(BG)}} \, = \, -352.8\mbox{\ \emph{meV}} \, < \, 0 \, .
        \end{equation}
	\end{example}
	
	\begin{example}[Combined potential]\label{ex-LJ-Morse}
		Let 
		\begin{equation}
			v_{\mathrm{M+LJ}}(\bx,z) \, = \, v_{\mathrm{Morse}}(|\bx|)v_{\mathrm{LJ}}(z) \, ,
		\end{equation}
		where $v_{\mathrm{Morse}}$ denotes a Morse potential with parameters $(E_0,\kappa,r_0)$, and $v_{\mathrm{LJ}}$ denotes a Lennard-Jones potential with parameters $(E_0',\sigma)=(1,\sigma)$. Then we obtain
		\begin{equation}
			m_{\mathrm{M+LJ}}(z) \, = \, \frac{3\pi E_0}{\kappa^2|z|^2} \big(e^{2\kappa r_0}-8e^{\kappa r_0}\big)\Big[-7\Big(\frac{\sigma}{|z|}\Big)^6+26\Big(\frac{\sigma}{|z|}\Big)^{12}\Big] \, .
		\end{equation}
		Now there is a transition parameter $r_0=\log 8 \, \kappa^{-1}$ above which $m_{\mathrm{Comb}}$ transitions from positive to negative at $|z|=\big(\frac{26}7\big)^{\frac16}\sigma\approx1.24\sigma$, and below which the converse transition occurs. With the above parametrization, we obtain, in the case of bilayer graphene,
        \begin{equation}
            m_{\mathrm{M+LJ}}^{\mathrm{(BG)}}\, = \, 11.6\mbox{\ \emph{keV}} \, > \, 0 \, .
        \end{equation}
	\end{example}
	
	\begin{figure}[htp]
		\centering
		\begin{tikzpicture}
			\begin{axis}[
				axis x line=center,
				axis y line=center,
				xtick={0.193,0.95,1,2,2.35,4,4.354},
				ytick={-1},
				xlabel={$|z|$},
				xticklabels = {$\textcolor{blue}{d_1^{\pm}}$,$\textcolor{red}{d_2^{\pm}}$,$r_1$,$\nicefrac2\kappa$,$\textcolor{blue}{d_1^{\mp}}$,$r_2$,$\textcolor{red}{d_2^{\mp}}$},
				x tick label style={font=\small,yshift={1.75*mod(\ticknum+1,2)*1em}},
				xlabel style={below right},
				ylabel style={above left},
				xmin=0,
				xmax=5.5,
				ymin=-2.5,
				ymax=5.5,
				legend pos=outer north east]
				\addplot [mark=none,samples=200,domain=-3.5:1.9] {1/(x-2)};
				\addlegendentry{$(\kappa|z|-2)^{-1}$}
				\addplot [mark=none,samples=200,blue,domain=-5.5:7] {exp(x-1)-1};
				\addlegendentry{$\exp(\kappa(|z|-r_1))-1$}
				\addplot [mark=none,red,samples=200,domain=-5.5:7] {exp(x-4)-1};
				\addlegendentry{$\exp(\kappa(|z|-r_2))-1$}
				
				\addplot [mark=none,black,samples=200,domain=2.1:7] {1/(x-2)}; 
				\addplot [mark=none,black,dashed,domain=-3.5:7] {-1};
				\addplot +[mark=none,black,dashed] coordinates {(2, -5.5) (2, 5.5)}; 
			\end{axis}
		\end{tikzpicture}
		\caption{Transition distances for the Morse potential are given by the intersection of the hyperbola with the respective exponential function}
		\label{fig-Morse-transition}
	\end{figure}
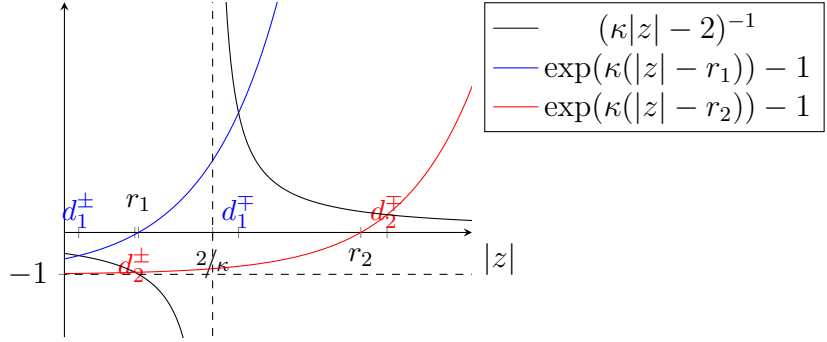
	
	In each of the examples, we determine regions where phonon instabilities are possible. In particular, we now assume that $m(d)<0$.
	
	\par Decomposing the interlayer energy according to \eqref{eq-inter-decomp} and employing \eqref{eq-D2v-av}, we obtain
	\begin{align}
		\Eapprox[\bv] \, :=& \, \sum_{j=1}^2\Big(\delta^2\demono_j(v_j)+ \frac\lambda2\int \dxx{x}\dx{y} D^2v(x-y,d) [(v_j(x)-v_j(y))^{\otimes 2}]\\
		& \quad + \, \frac\lambda2 (-1)^j\int \dxx{x}\dx{y} (v_j(x)-v_j(y))^T D^2v(x-y,d) (v_2(y)-v_1(y))\Big)\\
		& \quad + \, \lambda m(d)\int\dx{x}|(v_1(x)-v_2(x))_z|^2  \, .
	\end{align}
	In order to minimize the first line, we choose $v_j$ in each layer aligned. With that choice, the second line also vanishes in the limit. In order to minimize the last line, we choose $v_1$ and $v_2$ anti-parallel. In particular, choosing 
	\begin{equation}
		\bv_R(x) \, := \, \frac{\phi_R(x)}{(2\pi)^{\frac12}R}(e_z,-e_z) \, ,
	\end{equation}
	where we choose $\phi_R$ as an approximate normalized characteristic function of $B_R(0)$, i.e.
	\begin{align}
		\phi_R\Big|_{B_R(0)}\, = \, 1+O(R^{-1}) \, , \quad \phi_R\Big|_{B_{R+1}(0)^c}=0 \, , \quad \|\nabla\phi_R\|_\infty \lesssim 1 \, , 
	\end{align} 
	we obtain
	\begin{equation}
		\inf_{\|v_j\|_2=1}\Eapprox[\bv] \, \leq\, \limsup_{R\to\infty}\Eapprox[\bv_R] \, = \, 4\lambda m(d)\, < \, 0 \, ,
	\end{equation}
	for we chose $m(d)<0$.

	\subsection{Discrete model}
	
	Analogously to \eqref{eq-inter-decomp}, we decompose the interlayer energy as
	\begin{align}
		\begin{split}\label{eq-inter-decomp-discr}
			&\sum_{R_1,R_2} (v_1(R_1)-v_2(R_2))^T D^2v(R_1-R_2) (v_1(R_1)-v_2(R_2)) \, =\\
			& \quad \sum_{R_1,R_2} \Big((v_1(R_1)-v_1(\Lfloor{R_2}{1}))^T D^2v(R_1-R_2) (v_1(R_1)-v_1(\Lfloor{R_2}{1})) \\
			& \quad + (v_1(\Lfloor{R_2}{1})-v_2(R_2))^T D^2v(R_1-R_2) (v_1(\Lfloor{R_2}{1})-v_2(R_2))\\
			& \quad +2 (v_1(R_1)-v_1(\Lfloor{R_2}{1}))^T D^2v(R_1-R_2) (v_1(\Lfloor{R_2}{1})-v_2(R_2))\Big) \, .
		\end{split}
	\end{align}
	After finding a sequence of phonon modes with a limiting negative energy in the continuum approximation, we would like to study the functional in the true discrete case. Analogously to above, we now set
	\begin{equation}\label{eq-discrete-ansatz}
		v_1(R_1) \, := \,  \frac{\phi_R(R_1)|\Gamma_1|^{\frac12}}{(2\pi)^{\frac12}R}e_z \, , \quad v_2(R_2) \, := \, -\frac{\phi_R(R_2)|\Gamma_2|^{\frac12}}{(2\pi)^{\frac12}R}e_z \, .
	\end{equation}
	Plugging this ansatz into \eqref{eq-inter-decomp-discr}, we obtain
	\begin{align}
		\begin{split}
			&\sum_{R_1,R_2} (v_1(R_1)-v_2(R_2))^T D^2v(R_1-R_2) (v_1(R_1)-v_2(R_2)) \, =\\
			& \quad  \frac{(|\Gamma_1|^{\frac12}+|\Gamma_2|^{\frac12})^2}{2\pi R^2}\sum_{\substack{R_2\in \cR_2\cap B_R(0)\\R_1\in \cR_1}} {\partial_z^2}v(R_1-R_2) \, + \, O(R^{-1}) \, ,
		\end{split}
	\end{align}
	where $C_\phi$ is the normalization constant for $\phi$. Observe that we have that
	\begin{align}
		\sum_{\substack{R_2\in \cR_2\cap B_R(0)\\R_1\in \cR_1}} D^2v(R_1-R_2) \, &= \, \sum_{\substack{R_2\in \cR_2\cap B_R(0)\\R_1\in \cR_1}} D^2v(R_1-\Lfrac{R_2}{1}) \, .
	\end{align}
	In particular, employing the Birkhoff ergodic theorem, we obtain
	\begin{align}
		&\sum_{R_1,R_2} (v_1(R_1)-v_2(R_2))^T D^2v(R_1-R_2) (v_1(R_1)-v_2(R_2))\\
		&\quad = \, \frac{(|\Gamma_1|^{\frac12}+|\Gamma_2|^{\frac12})^2|\Gamma_2|}{2\pi|\Gamma_1|}\int_{\Gamma_1}\dx{x}\sum_{R_1\in\cR_1}\partial_z^2v(R_1-x) \, + \, O(R^{-1})\\
		&\quad = \, \frac{(|\Gamma_1|^{\frac12}+|\Gamma_2|^{\frac12})^2}{2\pi|\Gamma_1|/|\Gamma_2|} \int_{\R^2} \dx{x} \partial_z^2v(x) \, + \, O(R^{-1})\\
		&\quad = \frac{(|\Gamma_1|^{\frac12}+|\Gamma_2|^{\frac12})^2}{2\pi|\Gamma_1|/|\Gamma_2|}m(d) \, + \, O(R^{-1}) \, , 
	\end{align}
	where we used \eqref{eq-D2v-av}. Observe that the monolayer energy contributions are each of size $O(R^{-1})$. As a consequence and bounding the crossing term as above via Cauchy-Schwarz, we obtain
	\begin{equation}
		\inf_{\|\du_j\|_2=1}\detot(0)(\bdu,\bdu) \, \leq\, \lim_{R\to\infty}\Delta E_\infty(0)(\bv_R,\bv_R) \, \leq \, \lambda\frac{(|\Gamma_1|^{\frac12}+|\Gamma_2|^{\frac12})^2}{2\pi|\Gamma_1|/|\Gamma_2|}m(d) \, < \, 0 \, ,
	\end{equation}
	provided that $m(d)<0$.

	With analogous steps and following the ideas of \cite[Lemma 3.6]{hott2023incommensurate}, we obtain
	\begin{align}
		&\inf_{\substack{\|\du_j\|_2=1\\ j=1,2}}\detot(\bueq)(\bdu,\bdu) \\
		& \quad \leq \, \frac{(|\Gamma_1|^{\frac12}+|\Gamma_2|^{\frac12})^2}{2\pi|\Gamma_1|/|\Gamma_2|} \int_{\R^2} \dx{x} \partial_a^2 v\big(x+\ueq_2(\MoireRLV A_1^{-1}x)-\ueq_1(\MoireRLV A_2^{-1}x)\big)
	\end{align}
	for any unit vector $a\in\S^2$. With that, we obtain the following condition:
	
	\begin{definition}[Phonon instability]\label{defi-phonon-instable}
		Let $\bueq$ denote the equilibrium displacements. We say the bilayer system is \emph{phonon-instable} iff there exists a direction $a\in\S^2$
		\begin{align}\label{def-phonon-instable}
			\int_{\R^2} \dx{x} \partial_a^2 v\big(x+\ueq_2(\MoireRLV A_1^{-1}x)-\ueq_1(\MoireRLV A_2^{-1}x)\big) \, < \, 0 \, .
		\end{align}
	\end{definition}

	\paragraph{Sufficient condition for stability} Above, we showed a sufficient condition for instability. Conversely, we would like to establish a sufficient condition for stability. Using \eqref{eq-inter-decomp-discr}, we obtain positive (semi-) definiteness of $\delta^2\detot(0)$ if
	\begin{align}
		\sum_{R_1\in\cR_1} D^2v(x-R_1) \, \geq \, 0
	\end{align}
	is positive definite for all $x\in \Gamma_1$. Including relaxation, this condition becomes 
	\begin{align}
		\sum_{R_1\in\cR_1} D^2v\big(x-R_1+\ueq_2(\MoireRLV A_1^{-1}x)-\ueq_1(R_1+\MoireRLV A_2^{-1}x)\big) \, \geq \, 0 \, .
	\end{align}
	Since $\ueq_j$ is $\MoireRL$-periodic and
	\begin{equation}
		R_1 \, = \, -\MoireRLV (A_2^{-1}-A_1^{-1})R_1 \, = \, -\MoireRLV A_2^{-1}+\MoireRLV A_1^{-1}R_1 \, ,
	\end{equation}
	where $\MoireRLV A_1^{-1}R_1\in\MoireRL$, we have that
	\begin{align}\label{eq-u-def}
		&\ueq_2(\MoireRLV A_1^{-1}x)-\ueq_1(R_1+\MoireRLV A_2^{-1}x) \\
		&\quad = \, \ueq_2(\MoireRLV A_1^{-1}(x-R_1))-\ueq_1(\MoireRLV A_2^{-1}(x-R_1)) \, =: \, \ueq(x-R_1) \, .
	\end{align}
	With this, we obtain the following condition.
	
	\begin{definition}[Phonon stability]\label{defi-phonon-stable}
		Let $\bueq$ denote the equilibrium displacements, and $\ueq$ be defined as in \eqref{eq-u-def}. We say the bilayer system is \emph{phonon-stable} iff 
		\begin{align}\label{def-phonon-stable}
			\sum_{R_1\in\cR_1} D^2v(x-R_1+\ueq(x-R_1)) \, \geq \, 0 
		\end{align}
		for all $x\in\Gamma_1$, in the sense of quadratic forms.
	\end{definition}

	\paragraph{GSFE based offset energy} In order to find a model based on the \emph{generalized stacking fault energy} determined in \cite{Cazeaux-Massatt-Luskin-ARMA2020}, the interlayer energy needs to be replaced by
	\begin{align}
		\mavint\dx{x} \GSFE(\bot x+u_1(x)-u_2(x)) \, .
	\end{align}

	The corresponding interlayer phonon energy term is given by
	\begin{align}
		\sum_{R_1\in\cR_1} D^2\GSFE(\bot R_1+\ueq_1(R_1)-\ueq_2(R_1))\big[(\du_1(R_1)-\du_2(R_1))^{\otimes 2}\big] \, .
	\end{align}
	Analogously, one could replace the summation over $\cR_1$ by summation over $\cR_2$, or a combination of these. Notice that here, we would need to construct a smooth interpolant of $\du_2$. This could be avoided, if one instead approximated the energy by
	\begin{align}
		\sum_{R_1\in\cR_1} D^2\GSFE(\bot R_1+\ueq_1(R_1)-\ueq_2(R_1))\big[(\du_1(R_1)-\du_2(\Lfloor{R_1}{2}))^{\otimes 2}\big] \, ,
	\end{align}
	or, analogously, evaluating $\du_2$ at the nearest lattice point in $\cR_2$ instead. $\GSFE$ takes the form
    \begin{align}
        \GSFE(v,w) \, :=& \GSFE(2\pi v,2\pi w) \, ,\\
        \GSFE_0(v,w) \, :=& \, c_0 +c_1[\cos(v)+\cos(w)+\cos(v+w)]\\
        & \, +c_2[\cos(v+2w)+\cos(v-w)+\cos(2v+w)]\\
        & \, + c_3 [\cos(2v) + \cos(2w)+\cos(2v+2w)] \, ,
    \end{align}
	where $c_0=6.832$, $c_1=4.064$, $c_2=-0.374$, and $c_3=-0.095$ (all in meV/unit area), see \ref{fig-GSFE-ARMA}. 
    In particular, the stability criterion \ref{defi-phonon-stable} translates to
    \begin{align}
        D^2\GSFE(\bot x+\ueq_1(x)-\ueq_2(x)) \, \geq \, 0 \, .
    \end{align}
    
    \begin{figure}
        \centering
        \includegraphics[width=0.5\linewidth]{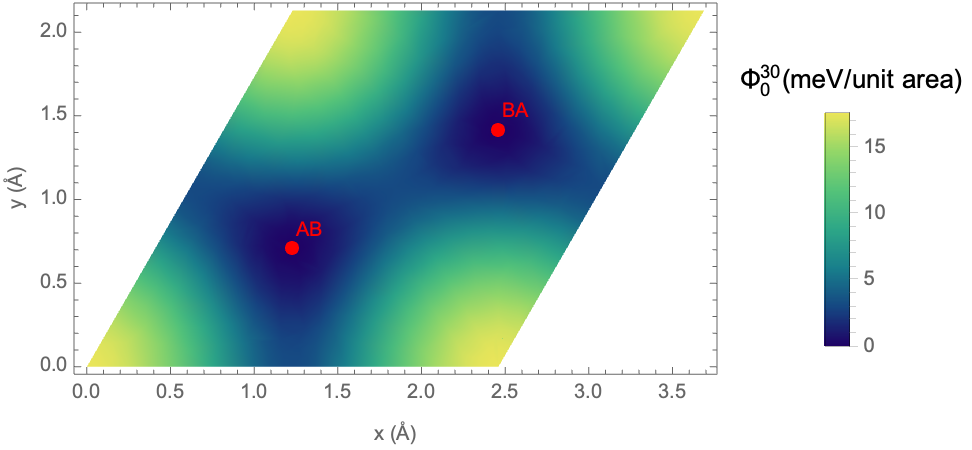}
        \caption{Profile of $\GSFE$, at $30^\circ$ rotated position, with potential minima AB/BA.}
        \label{fig-GSFE-ARMA}
    \end{figure}

    \section{Asymptotic analysis\label{sec-asymp-analysis}}

    To test our criterion, we shall restrict to purely twisted systems, i.e., we assume
    \begin{align}
        A_1 \, = \, R_{\theta/2}A \, , \quad A_2 \, = \, R_{-\theta/2}A \, .
    \end{align}
    Moreover, to obtain an analytical result, we study the energy in the limit $\theta\to0$. As we will see, one can reduce the functional to an Allen-Cahn functional.
    
    \subsection{Reduction to Allen-Cahn functional}
    
    For that, let us consider the energy in the Cauchy-Born and linear elasticity approximation and with interlayer pair potentials. In this case, the total energy reads, see \cite{hott2023incommensurate},
    \begin{align} \label{our-energy}
	   \frac12 \sum_{j=1}^2\mavint \dx{x} \nabla u_j(x):C_j:\nabla u_j(x) + \int_{\R^2} \frac{\dxx{\xi}}{|\MoireCell|} v\big((A_1-A_2) \xi + u_1(A_1\xi) -u_2(A_2\xi)\big) \, .
    \end{align}
    We rescale $x=\MoireRLV y$, $\xi=(A_1-A_2)^{-1}A y$ and abbreviate
    \begin{align}
     U_j(y) &:=A^{-1}u_j(\MoireRLV y), \quad \vep:=2\sin(\sfrac\theta2), \quad I_\vep := A^{-1}R_{\frac\theta2}A , \quad V(x):=\frac{v(Ax)}{|\det A|}, \\
     \cE_{j,abcd}&:=(A^{-1}R_{\frac\pi2})_{a'a}A_{b'b}(A^{-1}R_{\frac\pi2})_{c'c}A_{d'd}C_{j,abcd}\\
     &\,= \, \lambda(R_{-\frac\pi2}A^{-T}A)_{ab}(R_{-\frac\pi2}A^{-T}A)_{cd}+\mu\big[(R_{-\frac\pi2}A^{-T}A)_{ad}(A^TA^{-1}R_{\frac\pi2})_{bc}\\
     & \qquad +(R_{-\frac\pi2}A^{-T}A^{-1}R_{\frac\pi2})_{ac} (A^TA)_{bd}\big] \, .
    \end{align}
    Observe, that $U_j$ is $\Z^2$-periodic, $j=1,2$. After dividing by $\vep$, we rewrite the energy functional as
    \begin{align}
        \MoveEqLeft\frac\vep2\sum_{j=1}^2\int_{\T^2}\dx{y}\nabla U_j(y):\cE_j:\nabla U_j(y) \, + \, \frac1\vep\int_{\R^2}\dx{y}V\big(y+U_1(I_\vep^{-1}y)-U_2(I_\vep y)\big)\\
        &= \, \int_{\T^2}\dx{y}\Big(\frac\vep2\sum_{j=1}^2\nabla U_j(y):\cE_j:\nabla U_j(y)\\
        &\qquad +\frac1\vep\sum_{n\in\Z^2}V\big(y-n+U_1(I_\vep^{-1}(y-n))-U_2(I_\vep (y-n))\big)\Big) \, . \label{eq-resc-inter}
    \end{align}
    If we formally take the limit $\vep\to0$, we obtain that the interlayer misfit energy, multiplied by $\vep$, converges to
    \begin{align}
        \int_{\T^2}\dx{y}\sum_{n\in\Z^2}V\big(y-n+U_1(y-n)-U_2(y-n)\big) \, = \, \int_{\T^2}\dx{y}\Phi_0\big(y+U_1(y)-U_2(y)\big) \, ,
    \end{align}
    where we used the fact that $U_i$ is $\Z^2$-periodic and where $\Phi_0(y):=\sum_{n\in\Z^2}V(y-n)$ is the \emph{misfit energy}, akin to the \emph{generalized stacking fault energy} (GSFE) obtained in \cite{Cazeaux-Massatt-Luskin-ARMA2020}. It turns out that $\Phi_0$ is a two-well potential with minima exactly at AB and BA stacking, corresponding to the positions $x_{AB}=\frac13(1,1)^T$ and $x_{BA}=2x_{AB}$. For simplicity, let us assume homobilayers, i.e., $\cE_1=\cE_2$. Using the fact that
    \begin{align}
        \frac12\sum_{j=1}^2\nabla U_j:\cE:\nabla U_j \, = \, \nabla u:\cE:\nabla u \, + \, \nabla \bar{u}:\cE:\nabla \bar{u} \, ,
    \end{align}
    where $u:=\frac{U_1-U_2}2$ and $\bar{u}:=\frac{U_1+U_2}2$. In particular, we can rewrite the energy as 
    \begin{align}
        \int_{\T^2}\dx{x} \big(\vep\nabla u:\cE:\nabla u+\frac1\vep\Phi_0(x+2u) \, + \, \vep \nabla \bar{u}:\cE:\nabla \bar{u} \big) \, .
    \end{align}
    Consequently, the minimum is attained at $\bar{u}=C$ for some constant $C$. Due to the constraint $\int_{\T^2}\dx{x}U_j=0$, we have that $C=0$, i.e., $u=U_1=-U_2$. Then the remaining energy is given by
    \begin{equation}
        \int_{\T^2}\dx{x} \big(\vep\nabla u_\vep:\cE:\nabla u_\vep+\frac1\vep\Phi_0(x+2u_\vep) \, .
    \end{equation}
    From the Allen-Cahn theory, see, \cite{sternberg1991vector}, we find that, as $\vep\to0$, the sequence of minimizers $x+2u_\vep$ needs to converge to 
    \begin{align}\label{eq-minimizer-gamma-limit}
        x_{AB}\mathds{1}_E+x_{BA}\mathds{1}_{E^c} 
    \end{align}
    for some appropriate set $E\subseteq \T^2$. Here, we set the minima of $\Phi_0$ to zero. A crucial assumption for the theory to be applicable is that the Hessian of $\Phi_0$ is positive at its wells. In particular, we require the assumption
    \begin{align}\label{eq-asymp-phonon-stable-crit}
        (D^2\Phi_0)(x_{AB}),(D^2\Phi_0)(x_{BA}) \, > \, 0 \, .
    \end{align}
    In particular, we have in this case that
    \begin{align}
        (D^2\Phi_0)(x_{AB}\mathds{1}_E+x_{BA}\mathds{1}_{E^c}) \, > \, 0 \, ,
    \end{align}
    i.e., by Definition \ref{defi-phonon-stable}, in the limit $\theta\to0$, the bilayer system is \emph{phonon stable}.

    \subsection{Model evaluation\label{sec-model-eval}}

    We now want to evaluate the asymptotic stability criterion \eqref{eq-asymp-phonon-stable-crit} in the models presented in Section \ref{sec-inter-models}. For simplicity, we will restrict to two-dimensional interlayer coupling pair-potentials, where we keep the vertical distance fixed. In the case of bilayer graphene, the geometry is given by 
    \begin{equation}
    	A \, = \, \atd\begin{pmatrix}
    		\frac{\sqrt{3}}{2} & \frac{\sqrt{3}}{2} \\
    		-\frac{1}{2} & \frac{1}{2} \\
    	\end{pmatrix} , \quad \atd \, = \, \sqrt{3}\cdot 1.42 \mathrm{\ nm} \, \approx \, 2.46\mathrm{\ nm} \, , \quad  d_{\mathrm{BG}} \, \approx \, 3.35 \mbox{\ \AA} \, ,
    \end{equation}
    where $d_{\mathrm{BG}}$ is the equilibrium distance and $\atd/\sqrt{3}$ is the lattice constant.
    For a radial potential $v^{(3D)}(x,z)=f(|(\bx,z)|)$, the restricted potential is then given by 
    \begin{equation}
    	v^{(2D)}(x) \, = \, f(\sqrt{|x|^2+d_{\mathrm{BG}}^2}) \, =: \, g(|x|) \quad \Rightarrow \quad V(x) \, = \, \frac{v^{(2D)}(Ax)}{|\det A|} \, .
    \end{equation}
    Observe that we have that
    \begin{align}
		D^2v^{(2D)}(x) \, = \, g''(r) e_r\otimes e_r \, + \, \frac{g'(r)}r e_\vphi\otimes e_\vphi, \quad D^2V(x) \, = \, \frac{A^TD^2v^{(2D)}(Ax)A}{|\det A|} \, , 
    \end{align}
    where $x=re_r$, $e_r=(\cos\vphi,\sin\vphi)^T$, $e_\vphi=(-\sin\vphi,\cos\vphi)^T$. To approximate $\Phi_0$, we truncate the summation at some $N$, i.e., we compute
    \begin{equation}
    	(D^2\Phi_0)(AB) \, \approx \, \sum_{\substack{n\in \Z^2\\ \|n\|_\infty \leq N}} D^2V(AB) \, .
    \end{equation}
    We provided the profiles of the GSFEs corresponding to the respective models, and the evaluation of the Hessians at the potential wells in fig. \ref{fig-stab-eval}. We obtain that each of the considered models exhibits phonon stability. 
    \par In case of $\GSFE$, an analogous argument to above shows that the minimizer needs to satisfy $u_2=-u_1$. With analogous notation to above, the argument of $\GSFE$ is given by $x+2u_\vep$, which, again, needs to converge to \eqref{eq-minimizer-gamma-limit}. In particular, we compute 
    \begin{align}
        D^2\GSFE (AB) \, \geq \, 117 \mbox{\ meV},
    \end{align}
    i.e., the GSFE-based model obtained in \cite{Cazeaux-Massatt-Luskin-ARMA2020} is phonon-stable in the limit $\theta\to0$.
    \begin{figure}
    	\centering
    	\begin{subfigure}{0.4\textwidth}
    		\includegraphics[width=\textwidth]{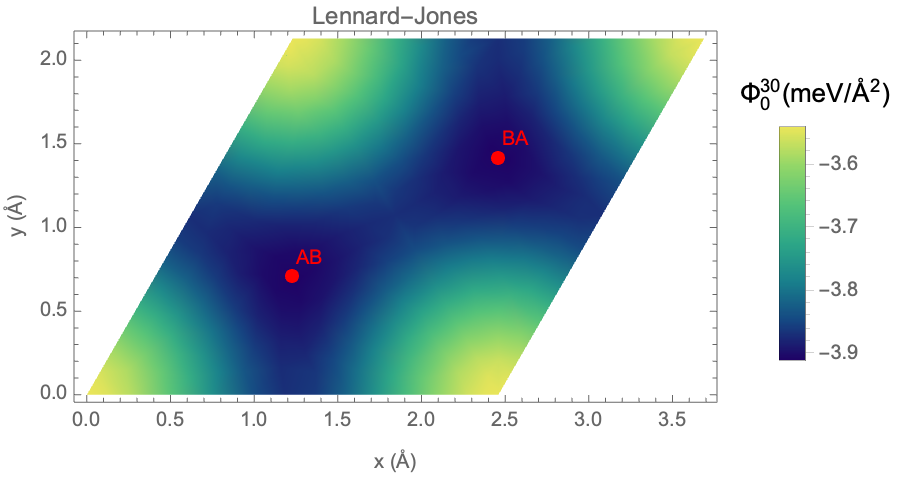}
    	\end{subfigure}
    	\begin{subfigure}{0.4\textwidth}
    		\includegraphics[width=\textwidth]{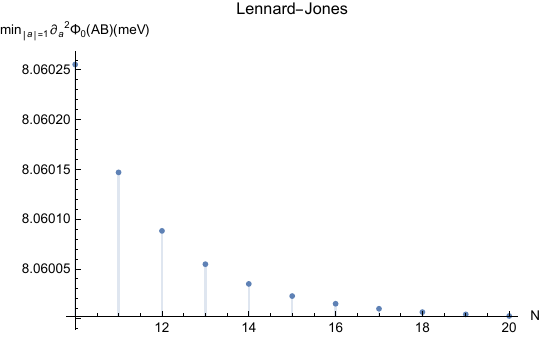}
    	\end{subfigure}
    	\hfill
    	\begin{subfigure}{0.4\textwidth}
    		\includegraphics[width=\textwidth]{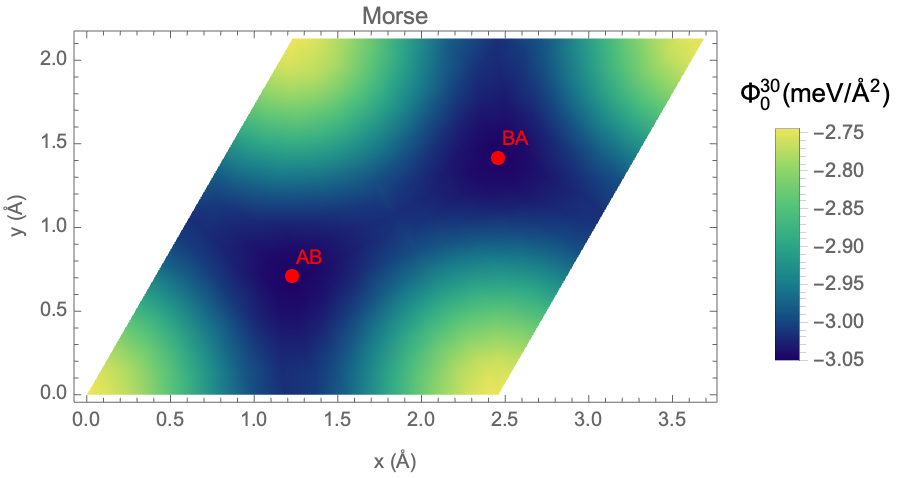}
    	\end{subfigure}
    	\begin{subfigure}{0.4\textwidth}
    		\includegraphics[width=\textwidth]{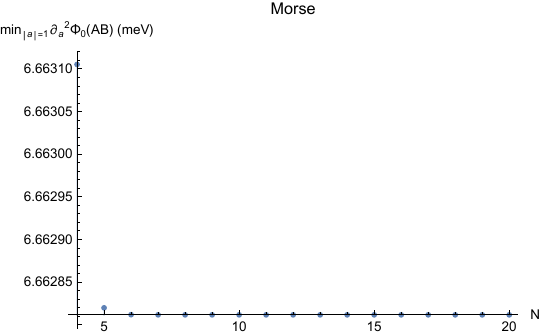}
    	\end{subfigure}
    	\hfill
    	\begin{subfigure}{0.4\textwidth}
    		\includegraphics[width=\textwidth]{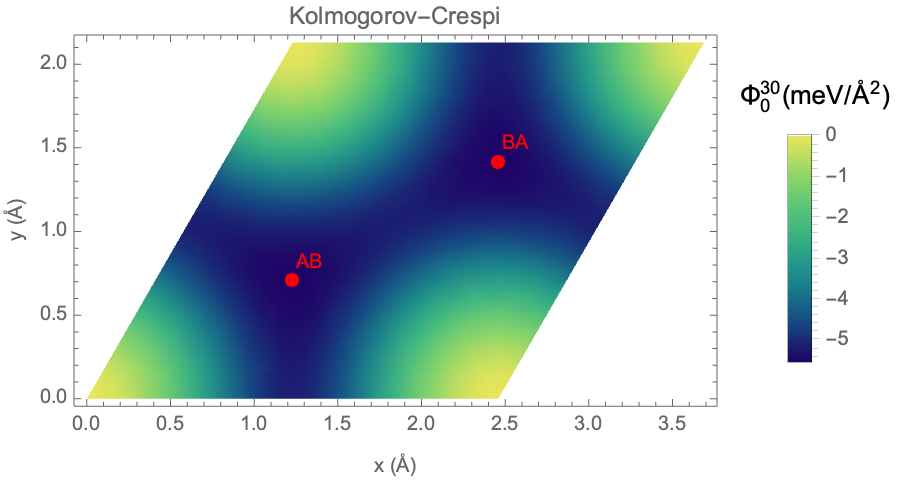}
    	\end{subfigure}
    	\begin{subfigure}{0.4\textwidth}
    		\includegraphics[width=\textwidth]{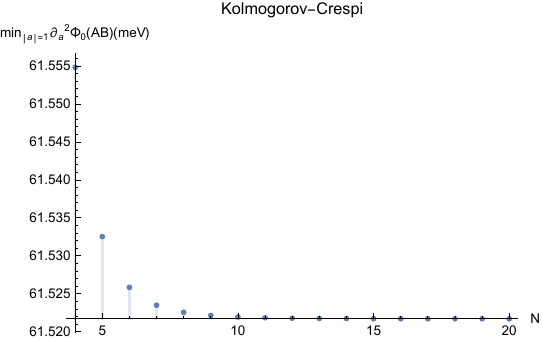}
    	\end{subfigure}
    	\caption{Misfit energy for bilayer graphene evaluated at 30$^\circ$ rotated positions with potential wells AB/BA. The truncation is chosen to be $N=10$.  Parametrization given in Sections \ref{sec-models} and \ref{sec-cont-approx}. \emph{Right:} Convergence of minimal eigenvalue of Hessian at AB/BA point for growing potential sum truncation $N$. From top to bottom: Lennard-Jones, Morse, and Kolmogorov-Crespi potential}
    	\label{fig-stab-eval}
    \end{figure}


	\bibliographystyle{acm}  
	\bibliography{references}

\end{document}